\documentclass[12pt,a4paper]{article}
\usepackage[latin1,utf8]{inputenc}
\usepackage{amsmath}
\usepackage{amsthm}
\usepackage{amsfonts}
\usepackage{amssymb}
\usepackage{thmtools}
\usepackage{thm-restate}
\usepackage{graphicx}
\usepackage{nicefrac}
\usepackage{fullpage}
\usepackage{bm}
\usepackage{xcolor}
\usepackage[linesnumbered,ruled]{algorithm2e}
\usepackage{tikz}
\usepackage{enumitem}

\definecolor{citec}{HTML}{324FDA} 
\definecolor{linkc}{HTML}{A0111A}
\definecolor{urlc}{HTML}{b55c87}
\usepackage[
    colorlinks,
    citecolor=citec,
    linkcolor=linkc,
    urlcolor=urlc,
    bookmarks = true,
    breaklinks,
]{hyperref}
\usepackage{cleveref}

\newtheorem{thm}{Theorem}[section]
\newtheorem{prop}[thm]{Proposition}
\newtheorem{lemma}[thm]{Lemma}
\newtheorem{remark}[thm]{Remark}
\newtheorem{cor}[thm]{Corollary}

\newtheorem{claim}[thm]{Claim}
\newtheorem{defi}[thm]{Definition}
\newtheorem{cnst}[thm]{Construction}
\newtheorem{example}[thm]{Example}

\newcommand{\Fq}{\mathbb{F}_q}
\newcommand{\Fp}{\mathbb{F}_p}

\newcommand{\cC}{\mathcal{C}}
\newcommand{\OurfieldSize}{200\cdot \left( \binom{n}{2k-1} \cdot (2k)!\right) ^2}

\newcommand{\bX}{{\bf X}}

\newcommand{\bfv}{{\bf v}}
\newcommand{\bfu}{{\bf u}}
\newcommand{\bdelta}{{\bm \delta}}
\newcommand{\RS}{\text{RS}_{n,k}(\alpha_1, \ldots, \alpha_n)}
\newcommand{\ed}[2]{\textup{ED}(#1, #2)}

\newcommand{\ba}{{\boldsymbol\alpha}}

\newcommand{\fix}[1]{#1}

\newcommand{\F}{\mathbb{F}}
\newcommand{\cR}{\mathcal{R}}

\DeclareRobustCommand{\REVdel}[1]{{\color{red}\hypersetup{allcolors=red}#1}}

\providecommand{\DIFdel}[1]{}
\providecommand{\DIFaddbegin}{}

\providecommand{\DIFdelFL}[1]{}

\begin{document}

	\title{Anonymous Shamir's secret-sharing via Reed-Solomon Codes Against Permutations, Insertions, and Deletions}
\author{
Roni Con\thanks{Department of Computer Science, Technion--Israel Institute of Technology. \href{mailto:roni.con93@gmail.com}{\texttt{roni.con93@gmail.com}} This work was supported by the European Union (DiDAX, 101115134). 
Views and opinions expressed are those of the author(s) only and do not necessarily reflect those of the European Union or the European Research Council Executive Agency. Neither the European Union nor the granting authority can be held responsible for them. }
}
\date{}
\maketitle
\begin{abstract}
    In this work, we study the performance of Reed-Solomon codes against an adversary who first permutes the codeword symbols and then performs insertions and deletions.
    This adversarial model is motivated by recent interest in fully anonymous secret-sharing schemes \cite{eldridge2024abuse,bishop2025fully}.
    A fully anonymous secret-sharing scheme has two key properties: first, the identities of the participants are not revealed before the secret is reconstructed; second, the shares of any unauthorized set of participants are uniform and independent.
    In particular, the shares of any unauthorized subset reveal no information about the identity of the participants who hold them.

    We begin by observing that Reed–Solomon codes, when robust against an adversary who permutes the codeword and then deletes symbols, can be used to construct fully anonymous gap-threshold secret-sharing schemes.
    We then show that there exist $[n,k]$ Reed--Solomon codes (over sufficiently large fields) that are robust against an adversary that arbitrarily permutes the codeword and then performs $n-2k+1$ insertions and deletions to the permuted codeword.
    This implies the existence of a $(k-1, 2k-1, n)$ gap-threshold secret-sharing scheme that is fully anonymous. 
    That is, any $k-1$ shares reveal nothing about the secret, \emph{and} no information on the participants’ identities.
    Conversely, any $2k-1$ shares suffice to reconstruct the secret without supplying the participants' identities to the reconstruction algorithm.
    We also provide explicit constructions of such schemes based on previous work on Reed--Solomon codes capable of correcting insertions and deletions.
    The constructions presented here are the first gap-threshold secret-sharing schemes to simultaneously achieve perfect uniform anonymity up to the privacy threshold and perfect anonymous reconstruction.
\end{abstract}
\thispagestyle{empty}

\newpage
\tableofcontents

\newpage

\section{Introduction}

Reed--Solomon (RS) codes are a crucial family of error-correcting codes that play a significant role in ensuring data integrity across various communication and storage systems. 
Developed in the 1960s \cite{reed1960polynomial}, these codes are particularly effective in correcting multiple symbol errors, making them ideal for data transmission and storage. Specific applications include distributed storage systems, QR codes, and data transmission over noisy channels.
The ubiquity of these codes is due to their beautiful algebraic structure, which allows us to efficiently encode and decode in the presence of substitutions and erasures.

Substitution and erasure errors have been extensively studied, and many ingenious constructions (including Reed--Solomon codes) are known that achieve optimal or near-optimal performance. 
By contrast, insertion and deletion (insdel) errors—where an insertion adds a symbol between two adjacent ones (or at the word boundaries) and a deletion removes a symbol—have received comparatively less attention. 
In recent years, however, motivated both by theoretical challenges and by practical applications in DNA-based storage, the design of insdel-correcting codes has become a vibrant research area. 
Comprehensive surveys on random insdel channels include Mitzenmacher \cite{mitzenmacher2009survey} and Cheragchi and Ribeiro \cite{cheraghchi2020overview}. For a survey about adversarial insdel errors and more related results, see \cite{haeupler2021synchronization} by Haeupler and Shahrasbi. For a detailed account of applications to DNA storage, see the survey by Sabary et al. \cite{sabary2024survey}.

Given the central role of Reed--Solomon codes in coding theory, a natural question is how well they perform against insdel errors, and indeed, this has been studied in a series of works \cite{safavi2002traitor,tonien2007construction,duc2019explicit,con2023reed,liu20212,con2023optimal,liu2024optimal,con2024random,beelen2025reed}. 
In particular, \cite{con2023reed} established that RS codes over sufficiently large fields achieve the so-called \emph{half-Singleton bound}, i.e., the optimal rate--error correction trade-off attainable by any \emph{linear} code.\footnote{Nonlinear codes are strictly stronger for insdel errors. 
Any linear code correcting a $\delta$-fraction of insdel errors has rate at most $\tfrac{1-\delta}{2} + o(1)$, whereas nonlinear constructions can approach the Singleton bound; for instance, \cite{haeupler2017synchronization} achieves rate close to $1-\delta$.} 
Nevertheless, the precise capabilities of RS codes, and more generally of linear insdel-correcting codes, remain far from fully understood; see \cite[Section~1.4]{con2024random} and \cite[Section~10]{cheng2020efficient}.

This paper concerns the Shamir secret-sharing scheme, one of the most notable applications of RS codes in cryptography. We make a novel connection between the anonymity of Shamir's secret-sharing and RS codes that are robust against an adversary that first permutes the codeword and then deletes symbols from the permuted codeword. 
\fix{The high-level intuition for this connection is as follows. In Shamir's scheme, the dealer distributes a codeword of an RS code among the players. 
If an adversary can arbitrarily permute the shares and delete some of them, 
then robustness against permutations and deletions ensures that the secret can still be reconstructed without supplying the reconstruction algorithm with the identity of the player associated with each share.
The anonymity property arises from a the following fact about $[n,k]_q$ RS codes (which holds for all MDS codes):
A restriction of the code to any $k$ columns is equals the space $\Fq^k$.}
Before describing in further detail the connection, we recall the notion of secret-sharing.

Secret-sharing, in general, is a cryptographic method that divides a secret (such as a cryptographic key or sensitive information) into multiple parts, known as shares, which are distributed among a group of participants. 
The key feature of secret-sharing is that only specific subsets of participants, defined by an \emph{access structure}, can reconstruct the original secret from their shares.\footnote{In this paper we shall focus only on the specific case of threshold access structures. That is, every set of participants of size at least $r$ (for some predetermined $r$) can reconstruct the secret).} One of the most well-known secret-sharing schemes in the literature is that of Shamir's \cite{shamir1979share}. 
The formal definition of Shamir's secret-sharing (SSS) is as follows. Given a secret $s\in \mathbb{F}$ to be distributed, a reconstruction threshold $k>0$ and $n$ parties (identified as $n$ distinct elements $\alpha_i\in \mathbb{F}$),  a dealer randomly selects a polynomial $f\in \mathbb{F}[x]$ of degree less than $k$ such that $f(0)=s$, and sends to party $i$ the share $f(\alpha_i).$  
It is easy to see that given $k$ pairs $(\alpha_i, f(\alpha_i))$, one can reconstruct the polynomial $f(x)$ (using interpolation) and therefore recover the secret $f(0)$. 
Moreover, any $k-1$ such pairs reveal nothing about the secret. 

In \cite{eldridge2024abuse,bishop2025fully}, the authors explored a stronger version of secret-sharing, where shares conceal information about the identities of the participants who hold them, and the secret is reconstructed without requiring the algorithm to know the identities of the share holders. 
Several notions of anonymity are discussed in \cite{bishop2025fully}, with the strongest one ensuring that any unauthorized set of shares is uniform and independent (such schemes are termed \emph{fully anonymous}).
In \cite{eldridge2024abuse}, the authors presented a variant of Shamir's secret-sharing scheme in which the shares conceal their identity, though not to the extent of being fully anonymous. 
One of the main limitations, aside from the lack of full anonymity, is that the reconstruction algorithm may fail (albeit with negligible probability, assuming sufficiently large share sizes).

In this work, we study the performance of RS codes against an adversary that permutes the codeword (arbitrary permutation) and then performs $t$ insdel errors. 
Our main motivation for this model is to construct a \emph{fully} anonymous Shamir secret-sharing scheme with perfect reconstruction. 
More specifically, we show the existence of $n$ distinct points $\alpha_i, i\in [n]$ over a large enough field such that for any polynomial, $f$, of degree $<k$, the following holds. 
For any set of $2k-1$ elements $I\subset [n]$, given the tuple $(f(\alpha_{I_1}), \ldots, f(\alpha_{I_{2k-1}}))$, one can reconstruct the polynomial $f$ and learn the secret $f(0)$. 
We emphasize that these evaluations are given in an arbitrary order and the index sequence $I$ is not supplied to the reconstruction algorithm. Thus, reconstruction does not require the identity associated with any input share. This does not guarantee that identities remain hidden after reconstruction: the recovered polynomial may identify the participants who supplied the shares.

We note that given any set $I\subset [n]$ of size $k-1$, the tuple $(f(\alpha_{I_1}), \ldots, f(\alpha_{I_{k-1}}))$ reveals nothing about the secret. Moreover, since these are $k-1$ evaluations of a random polynomial $f$ of degree $\leq k-1$, this tuple is uniformly distributed over $\Fq^{k-1}$, which guarantees the strongest form of anonymity.
This guarantee applies to at most $k-1$ distinct shares from one sharing. It also holds jointly for fixed collections of at most $k-1$ distinct shares from each of several independently randomized sharings. The secrets may coincide, but the dealer must use fresh independent randomness for each sharing.

The scheme presented in this paper is the \emph{first} gap-threshold scheme that satisfies the ramp analogue of the strongest notion of anonymity, namely, (1) the shares of every set of at most the privacy threshold are \emph{perfectly} uniform and independent for every fixed secret, and (2) reconstruction is \emph{perfect} and does not require the identities of the reconstructing parties.

We note that in \cite{bishop2025fully}, the authors explicitly asked if there are threshold (and also general access structure) secret-sharing schemes that are fully anonymous and have perfect reconstruction. Our results provide an affirmative answer for the gap-threshold case, while the exact threshold case remains open.  

\subsection{Setup and Preliminaries}

For an integer $k$, we denote $[k]=\{1,2,\ldots,k\}$. 
	Throughout this paper, $\log(x)$ refers to the base-$2$ logarithm. For a prime power $q$, let $\F_q$ denote the finite field of size $q$.

	A linear code over a field $\F$ is a linear subspace $\cC\subseteq \F^n$. The rate of a linear code $\cC$ of block length $n$ is $\cR=\dim(\cC)/n$. Every linear code of dimension $k$ can be described as the image of a linear map, which, by abuse of notation, we also denote with $\cC$, i.e., $\cC : \F^k \rightarrow \F^n$.
    When $\cC\subseteq \F_q^n$ has dimension $k$ we say that it is an $[n,k]_q$ code (or an $[n,k]$ code defined over $\Fq$). The \fix{minimum distance} of a code $\cC$ with respect to a metric $d(\cdot,\cdot)$ is defined as $\min_{\bfv, \bfu \in \cC,\bfu \neq \bfv}{d(\bfv,\bfu)}$. 
    We denote by $\Fq[X_1, \ldots, X_n]$ the ring of polynomials in $X_1,\dots,X_n$ over $\Fq$ and by $\Fq(X_1, \ldots, X_n)$ the field of fractions of that ring. We write $\Fq^{<k}[X]$ to refer to all univariate polynomials of degree at most $k-1$ and coefficients in $\Fq$. 

    \begin{defi}
        We call a sequence $I\in [n]^{\ell}$ a \emph{distinct-element sequence} if all of its elements are pairwise distinct; that is, $I_i \neq I_j$ for all $i\neq j\in [\ell]$.
    \end{defi}

    We use capital letters (e.g., $I, J$) to denote sequences or sets; the intended meaning will be clear from context.
    \begin{defi}
        We say that two distinct-element sequences are \emph{equal as sets} if they contain the same elements.
    \end{defi}

     \begin{example} \label{ex:ordered}
        The following two distinct-element sequences $I=(1,2,4)$ and $J=(4,1,2)$ are equal as sets.
    \end{example}
    Observe that, from a \emph{set} $I$, we can construct $|I|!$ distinct-element sequences that are pairwise equal as sets.                    
    We next define Reed--Solomon codes (RS codes). 

\begin{defi}[Reed--Solomon code]
    Let $\alpha_1, \alpha_2, \ldots, \alpha_n$ be distinct elements \fix{in $\mathbb{F}_q$}. For $k<n$, the $[n,k]_q$ \emph{Reed--Solomon} (RS) code of dimension $k$ and block length $n$ associated with the evaluation vector $(\alpha_1, \ldots, \alpha_n) \in \F_q^n$ is defined as the set of codewords 
    \[
    \text{RS}_{n,k}(\alpha_1, \ldots, \alpha_n) := \left \lbrace \left( f(\alpha_1), \ldots, f(\alpha_n) \right) \mid f\in \mathbb{F}_q[x],\text{ }\deg f < k \right \rbrace.
    \]	
\end{defi}

	That is, a codeword of an $[n,k]_q$ RS code is the evaluation vector of a polynomial of degree less than $k$ at $n$ predetermined, distinct points. 
	It is well known (and easy to verify) that the rate of an $[n,k]_q$ RS code is $k/n$, and the minimal distance, with respect to the Hamming metric, is $n - k + 1$.

    Throughout, we use the shorthand notations $\bX$ and $\ba$ to refer to $(X_1, \ldots, X_n)$ and $(\alpha_1, \ldots, \alpha_n)$, respectively.

\paragraph{Correcting Insertions, Deletions, and Permutations}

    Let $s$ be a string over the alphabet $\Sigma$. The operation of removing a symbol from $s$ is called a \emph{deletion}, and the operation of inserting a new symbol from $\Sigma$ either between two consecutive symbols in $s$, at the beginning, or at the end of $s$ is called an \emph{insertion}.

     We recall the notions of a subsequence and a longest common subsequence. 
    \begin{defi}
            A \emph{subsequence} of a string $s$ is a string obtained by removing some (possibly none) of the symbols in $s$. 
    \end{defi}
    \begin{defi}

        Let $s$ and $s'$ be strings over an alphabet $\Sigma$. 
        A \emph{longest common subsequence} between $s$ and $s'$, is a subsequence of both $s$ and $s'$ of maximal length. We denote the length of a longest common subsequence by $ \textup{LCS}(s, s')$.

        The \emph{edit distance} between $s$ and $s'$, denoted by $\ed{s}{s'}$, is the minimal number of insertions and deletions needed to turn $s$ into $s'$. 
    \end{defi}
    A well-known fact (see, e.g., \cite[Lemma 12.1]{crochemore2003jewels}) states that $\ed{s}{s'} = |s| + |s'| - 2\textup{LCS}(s,s')$, and this implies that the insdel correction capability of a code is determined by the LCS of its codewords. Specifically,
    \begin{lemma} \label{lem:code-lcs}
        A code $C$ can correct $\delta n$ insdel errors if and only if $\textup{LCS}(c,c')\leq n - \delta n - 1$ for any distinct $c,c'\in C$.
    \end{lemma}

    In this paper, we study the performance of RS codes against an adversary that first applies an arbitrary permutation to a codeword and then performs insdel errors. Formally,
    \begin{defi}
        The \emph{$t$-permutation-insdel} adversary is defined as an adversary that, given a codeword $c\in \cC \subseteq \Sigma ^n$, chooses an arbitrary permutation on $n$ elements $\pi$, and then
        \begin{enumerate}
            \item Applies $\pi$ on $c$, i.e., $\pi(c) := (c_{\pi(1)}, \ldots, c_{\pi(n)})$.
            \item Performs $t$ insdel errors to $\pi(c)$ and outputs the result.
        \end{enumerate}
    \end{defi}
    \begin{defi}
        A code is said to be \emph{robust against the $t$-permutation-insdel adversary} if the adversary cannot produce the same output on two distinct codewords.
    \end{defi}
    \begin{remark}
        Kova\v{c}evi\'c and Tan \cite{kovavcevic2018codes} comprehensively studied codes capable of correcting permutations, insertions, deletions, and substitutions. They provide constructions and bounds on the coding parameters, focusing on regimes in which $q$, the field size, is either constant or linear in $n$. 
        The adversarial model considered in this work clearly falls within their framework. 
        The metric they define for this model is that of distance between ``histograms.'' More specifically, since any permutation can be applied to a codeword, the order of its symbols does not matter. 
        Thus, a codeword is viewed as a histogram of its values over $\Fq$ and then, the $L_1$ distance between histograms of two codewords is the minimal number of insertions and deletions needed to transform one codeword into a permuted version of the other.

        For convenience and to be consistent with the previous works that studied RS codes against insdel errors, we shall work with the permutation-insdel adversary. This allows us to formulate algebraic conditions similar to those in \cite{con2023reed,con2024random}.
        However, arguably the more ``natural'' way to describe such codes is using the notations and general framework presented in \cite{kovavcevic2018codes}. 
    \end{remark}

        \fix{For a distinct-element sequence $I = (I_1,\ldots,I_t)$ and a vector $\bfv = (v_1,\ldots, v_n) \in \Fq^n$, we denote $\bfv_I = (v_{I_1}, \ldots, v_{I_t})$}. We begin with the following equivalent characterization of a code that is robust against the $t$-permutation-insdel adversary.
       \begin{claim} \label{clm:correcting-cond}
        A code $\cC$ is robust against the $t$-permutation-insdel adversary if and only if, for all distinct $c,c'\in \cC$ and any two distinct-element sequences $I, J\in [n]^{n-t}$, we have $c_I \neq c'_J$.
    \end{claim}
    \begin{proof}
        $(\Rightarrow)$. Suppose that there exist two distinct codewords $c,c'\in \cC$ and distinct-element sequences $I, J\in [n]^{n-t}$ such that 
        \[
        c_I = (c_{I_1}, \ldots, c_{I_{n-t}}) = (c'_{J_1}, \ldots, c'_{J_{n-t}}) = c'_J\;.
        \]
        Let $\pi\colon [n] \to [n]$ be a permutation such that $\pi(i) = I_i$ for all $i\in [n-t]$. Similarly, let $\pi' \colon [n]\to [n]$ be a permutation such that $\pi'(i) = J_i$, for all $i\in [n-t]$. Thus, if we apply $\pi$ to $c$, $\pi'$ to $c'$, and delete the last $t$ elements of each, we achieve the same output, in contradiction.

        \noindent$(\Leftarrow)$. Assume that the code is not robust against the $t$-permutation-insdel adversary. Then, there exist two distinct codewords $c$ and $c'$ and two permutations $\pi$ and $\pi'$ such that $\ed{\pi(c)}{\pi'(c')} \leq 2t$. 
        Thus, $\pi(c)$ and $\pi'(c')$ have a common subsequence of length at least $n-t$. Choose one of length exactly $n-t$, and let $I,J\in[n]^{n-t}$ be its increasing sequences of positions in $\pi(c)$ and $\pi'(c')$, respectively.
        Set $I'_s=\pi(I_s)$ and $J'_s=\pi'(J_s)$ for every $s\in[n-t]$. 
        Since $\pi$ and $\pi'$ are permutations, $I'$ and $J'$ are distinct-element sequences, and
        \[
            c_{I'}=(\pi(c))_I=(\pi'(c'))_J=c'_{J'}\;.
        \]
        We arrived to a contradiction.
    \end{proof}

\paragraph{Anonymous Shamir's secret-sharing}

We begin with a brief reminder of threshold secret-sharing schemes. We refer the interested reader to Beimel's comprehensive survey \cite{beimel2011secret}. The following notations and definitions are inspired by \cite{kurihara2008new}.

Let $\mathcal{P} = \{P_i \mid 1\leq i \leq n \}$ be a set of $n$ participants. Let $\mathcal{D} \notin \mathcal{P}$ be a dealer who selects a secret $s\in \F_q$ and distributes a share $w_i \in \mathcal{W}_i$ to every participant, where $\mathcal{W}_i$ denotes the set of possible shares that $P_i$ can have. 
Also, assume that $u = | \mathcal{W}_i|$ for all $i\in [n]$.
The \emph{access structure}, typically denoted as $\Gamma \subset 2^{\mathcal{P}}$, is the set of all \emph{authorized sets}, i.e., sets that can recover the secret.
Any set that is not in $\Gamma$ is called an \emph{unauthorized set}.
A secret-sharing scheme is called an \emph{$(r,n)$ threshold secret-sharing scheme} if $\Gamma = \{A \subset \mathcal{P} | |A| \geq r \}$. That is, every set of at least $r$ participants can recover the secret.
Further, we want that for each $A\notin \Gamma$, the respective parties cannot learn anything about the secret. We denote by $\mathtt{Share}$ a randomized algorithm that gets a secret $s$ and produces a vector $\mathrm{sh}\in \mathcal{W}_1 \times \ldots \times \mathcal{W}_n$ and we denote by $\mathtt{Recon}$ the reconstruction algorithm of the scheme.

More formally, let $S$ and $W_i$ be the random variables induced by $s$ (the secret) and $w_i$ (the $i$-th share), respectively; we want 
\begin{equation} \label{eq:perfect-ss}
H(S|\mathcal{V}_A) = \begin{cases}
    0, \qquad |A| \geq r \\
    H(S), \,|A| < r
\end{cases}\;,
\end{equation}
where $H(\cdot)$ denotes the entropy of a random variable and $\mathcal{V}_A = \{W_i \mid i\in A\}$. \footnote{In this paper (as in Shamir's secret-sharing scheme), the secret $s$ is chosen uniformly at random from $\Fq$ and then $H(S) = \log_2 q$.}
In words, \eqref{eq:perfect-ss} implies that any set of $r$ parties or more can reconstruct the secret (there is no entropy in $H(S|\mathcal{V}_A)$) but any set of less than $r$ parties cannot reveal any information about the secret.
Note that the standard Shamir secret-sharing scheme achieves \eqref{eq:perfect-ss}.
Indeed, in Shamir's secret-sharing scheme, the dealer who wants to distribute a secret $s\in \F_q$ chooses a random polynomial $f$ of degree $<r$ such that $f(0) = s$. Then, the party $P_i$ gets the share $(\alpha_i, f(\alpha_i))$ where $\alpha_1, \ldots, \alpha_{n}$ are distinct, predetermined nonzero elements of a finite field $\F_q$.
By Lagrange interpolation, \fix{given} any $r$ such pairs $(\alpha_i, f(\alpha_i))$, one can reconstruct $f$ and thus know $f(0)$. 
However, for $r-1$ pairs, the probability that a specific value $s$ was the secret is $1/q$ (assuming that $s$ was a uniformly random secret from $\Fq$), thus an adversary who sees $r-1$ pairs learns nothing about the secret.

The scheme we construct in this paper is a \emph{ramp} threshold secret-sharing scheme (also called a gap-threshold scheme). That is, there is a gap between the number of parties that learn nothing about the secret and the number of parties that can reconstruct the secret.
\begin{defi}
    Let $t<r<n$ be integers. Given the above definitions, a $(t,r,n)$ \emph{ramp threshold secret-sharing scheme} is a secret-sharing scheme such that \begin{enumerate}
        \item For any $A\subset \mathcal{P}$ where $|A| \geq r$, it holds that $H(S|\mathcal{V}_A) = 0$.
        \item For any $A\subset \mathcal{P}$ where $|A| \leq t$, it holds that $H(S|\mathcal{V}_A) = H(S)$.
    \end{enumerate}
\end{defi}
\fix{For $t < |A| < r$, the conditional entropy $H(S|\mathcal{V}_A)$ may take any value between $0$ and $H(S)$ (including both ends), reflecting that subsets of intermediate size may reveal partial information, no information and even all the information about the secret.}
In the standard Shamir secret-sharing scheme, $r = t-1$. We note that when $A\subset \mathcal{P}$ is of size larger than $t$ but smaller than $r$, then we have no restriction on $H(S\mid \mathcal{V}_A)$, namely, it can be any value between $0$ and $H(S)$ (including both ends). 

We now define the anonymity of a ramp threshold secret-sharing scheme.
For a $(t,r,n)$ ramp scheme, we distinguish the family of \emph{privacy sets}, $\{A\subseteq\mathcal{P}:|A|\leq t\}$, from the family of \emph{reconstruction sets}, $\{A\subseteq\mathcal{P}:|A|\geq r\}$. Our definition is the ramp analogue of perfect U-anonymity in \cite{bishop2025fully}: for every fixed secret, the shares of every privacy set are independent and uniform over the common share domain. We impose no privacy or anonymity requirement on sets of intermediate size.
Note that we assume all shares come from the same domain; otherwise, a share can reveal information about the party who holds it (in particular, the standard Shamir secret-sharing scheme is not anonymous, as each share contains $\alpha_i$ which is uniquely associated with $P_i$). \fix{We denote the common share domain by $\Omega$, and we assume that the secret is a random element in $\Fq$.}

\begin{defi} \label{defi:full-anonymity}
    Given the definition above, a \emph{fully anonymous $(t,r,n)$ ramp secret-sharing scheme} is a secret-sharing scheme such that the following two conditions are satisfied:
    \begin{enumerate}
        \item \label{item:full-anon-share}(Share anonymity) for every secret $s\in \F_q$, every integer $\ell \leq t$, every distinct-element sequences $I\in [n]^{\ell}$, and every $\bfv\in \Omega^{\ell}$, we have 
        \[
            \Pr_{\mathrm{sh} \leftarrow \mathtt{Share}(s)}[\mathrm{sh}_I = \bfv] = |\Omega|^{-\ell}\;.
        \] 
        \item \label{item:full-anon-recon}(Perfect anonymous reconstruction) There exists a function $\mathtt{Recon}$ (possibly non-efficient) such that for every secret $s\in \Fq$, every distinct-element sequence $A$ of size $r$, and every $\mathrm{sh}$ that can be generated by $\mathtt{Share}(s)$, it holds that $\mathtt{Recon}(\mathrm{sh}_A) = s$. 
        We emphasize that $\mathtt{Recon}$ always succeeds and does not require the identities of the participants who hold the shares.
        \footnote{We require information-theoretic reconstruction. In \cite{eldridge2024abuse}, the authors provide an efficient algorithm that is anonymous.
        In this work, the reconstruction algorithm is a brute-force one and we leave as an intriguing open question to develop an efficient reconstruction algorithm.}
\end{enumerate} 
\end{defi}

    \begin{remark}
         It is often assumed that the identities of participants in Shamir's secret-sharing schemes are not a part of the shares.
         In this case, the shares of every set of at most the privacy threshold satisfy condition~\ref{item:full-anon-share} in \Cref{defi:full-anonymity}.
         However, in this case, the standard interpolation-based reconstruction algorithm requires each participant to provide not only $f(\alpha_i)$ but also $\alpha_i$. 
         This violates condition~\ref{item:full-anon-recon} in \Cref{defi:full-anonymity} the reconstruction function $\mathtt{Recon}$ must operate solely on the shares, without any auxiliary information about the participants.
    \end{remark}

\subsection{Previous work} \label{sec:prev-results}
\paragraph{Reed--Solomon codes against insdel errors.}
    We begin by citing the (non-asymptotic version of the) rate--error-correction tradeoff for \emph{linear} codes correcting insdel errors.
    \begin{thm}[Half-Singleton bound] \cite[Corollary 5.2]{cheng2020efficient}
        An $[n,k]_q$ linear code can correct at most $n - 2k + 1$ insdel errors.
    \end{thm}

    The performance of RS codes against insdel errors was first considered in \cite{safavi2002traitor}, in the context of traitor tracing. In~\cite{wang2004deletion}, the authors constructed a $[5,2]_q$ generalized RS code that corrects a single insdel error. In \cite{tonien2007construction}, an $[n,k]$ generalized RS code was constructed that is capable of correcting $\log_{k+1}n - 1$ insdel errors. 
    Several constructions of two-dimensional RS codes correcting $n-3$ insdel errors (the maximal number of insdel errors an $[n,2]_q$ code can correct) are given in~\cite{duc2019explicit,liu20212,con2023reed,con2023optimal}, where the best known field size, $O(n^3)$, is achieved in~\cite{con2023optimal}. Moreover, it has been shown that such codes must have $q = \Omega(n^3)$ \cite{con2023reed}, and thus we have a clear picture for optimal two-dimensional RS codes correcting insdel errors.

    For $k>2$, much less is known. It was shown in \cite{con2023reed} that over large enough fields, there exist RS codes that exactly achieve the half-Singleton bound. 
    Specifically, 
    \begin{thm}[{\cite[Theorem 16]{con2023reed}}] \label{thm:con-rs-res}
        Let $n$ and $k$ be positive integers such that $2k - 1 \leq n$. For $q = O\left(\binom{n}{2k-1}^2 \cdot k^2 + n^2\right)$, there exists an $[n,k]_q$ RS code that can \fix{correct} $n-2k+1$ insdel errors.
    \end{thm}
    Later,~\cite{con2024random} showed that there exist RS codes over linear-sized alphabets that almost attain the half-Singleton bound. More specifically, for $q\ge n+ 2^{\mathrm{poly}(1/\epsilon)}k$, with high probability, a random $[n,k]_q$ RS code corrects at least $ (1-\epsilon)n-2k+1$ insdel errors.

    Both results establish the existence of RS codes capable of correcting insdel errors. Explicit constructions for general $k$ achieving the half-Singleton bound are given in~\cite{con2023reed,liu2024optimal} over much larger ($\approx n^{k^k}$) fields.
    An interesting open question is whether one can provide explicit constructions (with identified evaluation points) of RS codes over these fields that achieve (or get close to) the half-Singleton bound.

\paragraph{Anonymous secret-sharing}
The notion of anonymous secret-sharing studied in several prior works \cite{stinson1988combinatorial,phillips1992strongly,blundo1997anonymous,kishimoto2002bound} focuses solely on anonymous reconstruction. That is, to reconstruct the secret, the algorithm requires only the shares, not the identities of the participants.

In a recent work, Eldridge et al.~\cite{eldridge2024abuse} also considered achieving anonymity, which they term as ``unlinkability'', in the following sense. Given multiple unauthorized sets of shares -- corresponding to several shared secrets -- an adversary cannot ``link'' any of the shares to one another. In particular, the adversary cannot determine whether two shares from different unauthorized sets belong to the same participant. 
They motivated this definition by showing how such secret-sharing schemes can be used to construct protocols to defend against attackers who use location-tracking devices such as AirTags to stalk a victim. 
Their protocols balance between the privacy requirements of such devices (ensuring that only the owner can access the location of their belongings) and the need to protect individuals from such tracking attacks (detecting whether a device has been secretly attached to a person). For further discussion of applications, we refer the reader to the excellent introduction in \cite{eldridge2024abuse}. 

To motivate our work, we briefly describe \cite[Construction $1$]{eldridge2024abuse}. 
In \cite{eldridge2024abuse}, each time a dealer shares a secret, they choose randomly $\gamma_1,\ldots,\gamma_n$ and sets the $i$-th share to be $(\gamma_i, f(\gamma_i))$.
This construction satisfies the unlinkability property that \cite{eldridge2024abuse} aimed for. However, this scheme has two disadvantages. 

First, it does not satisfy condition~\ref{item:full-anon-share} in \Cref{defi:full-anonymity}. 
For example, consider the case of reconstruction threshold 3, where in this case $\Omega = \Fq \times \Fq$. Consider $\bfv = \{(\gamma, \delta), (\gamma, \delta ')\}$, where $\delta \neq \delta'$. Clearly, the probability of generating two such shares for a given secret is $0$.

Second, this scheme does not give perfect reconstruction. Note that two shares may receive the same $\gamma$, which occurs with probability $1/q$.   
Since the reconstruction algorithm is simply polynomial interpolation that requires distinct pairs, to make the failure probability negligible, $q$ must be super-polynomial in $n$, which implies that the share size is $\omega(\log(n))$ bits rather than $\Theta(\log(n))$.
In particular, regardless of the field size, their reconstruction is not a perfect reconstruction, as there is always a positive probability of failure.

The study of anonymous secret-sharing for general access structures, where anonymity is also considered for unauthorized sets (and not only when reconstructing), was first studied in \cite{guillermo2003providing,paskin2020cryptographic}. 
Then, in~\cite{bishop2025fully}, the authors systematically studied several notions of anonymous secret-sharing schemes for general access structures (extending and generalizing the works \cite{paskin2020cryptographic,eldridge2024abuse}) in the information-theoretic and computational regimes. 
They provide explicit constructions as well as exponential lower bounds on the share size.

The notions considered in~\cite{eldridge2024abuse,bishop2025fully} include the following relaxations of perfect uniformity for privacy sets and perfect anonymous reconstruction:
\begin{enumerate}
    \item Allowing imperfect uniformity of the shares (e.g., shares of an unauthorized set are only statistically close to uniform) or imperfect reconstruction (namely, reconstruction fails with small probability).
    \item Replacing the uniformity requirement of the scheme (condition~\ref{item:full-anon-share} in \Cref{defi:full-anonymity}) with a weaker requirement on the shares. 
    We refer the reader to the definitions of Single-dealer Share Anonymity (S-anonymity) and Multi-dealer Share Anonymity (M-anonymity) in \cite{bishop2025fully}.
\end{enumerate}

Our goal is to construct a ramp scheme in which the shares of every privacy set are perfectly uniform and independent for every fixed secret, and reconstruction from every reconstruction set is perfect and requires no participant identities as input. 

\subsection{Our results}
    We prove the existence of an RS code over a sufficiently large field that is robust against the permutation-insdel adversary. Specifically,
    \begin{restatable}{thm}{mainres}\label{thm:main}
        Let $k\geq 2$ and $n$ be integers such that $2k - 1 \leq n$. Let $q$ be a prime such that $q \geq \OurfieldSize$. Then, there exists an $\RS$ code that is robust against the $(n-2k + 1)$-permutation-insdel adversary.
    \end{restatable}

    As a corollary, we obtain the existence of the following fully anonymous ramp secret-sharing scheme.
    This gives the ramp analogue of a perfect U-FASS scheme in the terminology of~\cite{bishop2025fully}.
    \begin{cor} \label{cor:ano-sec-ramp}
        Let $k$ and $n$ be integers such that $2k - 1 \leq n$. There exists a fully anonymous $(k-1, 2k - 1, n)$ ramp secret-sharing scheme where the secret and the shares are elements of $\Fq$ for any prime $q$ such that $q \geq \OurfieldSize$.
    \end{cor}

    \begin{proof}
        Let $\alpha_1, \ldots, \alpha_n$ be nonzero elements such that $\RS$ is robust against the $(n - 2k +1)$-permutation-insdel adversary.
        Consider the following scheme. Let $s\in \Fq$ be a secret, and let $f\in \Fq[X]$ be a random polynomial of degree $<k$ such that $f(0) = s$. 
        Then the $i$-th participant gets the share $f(\alpha_i)$.
        We now verify that the conditions of \Cref{defi:full-anonymity} for a fully anonymous ramp scheme are satisfied. 
        We begin with Condition~\ref{item:full-anon-recon}. Since $\RS$ is robust against the $(n - 2k + 1)$-permutation-insdel adversary, by \Cref{clm:correcting-cond}, for every distinct-element sequence $I \in [n]^{2k-1}$, the tuple $(f(\alpha_{I_1}), \ldots, f(\alpha_{I_{2k-1}}))$ suffices to recover $f$ and thus also the secret $f(0)$. 

        We now show that the share anonymity condition (Condition~\ref{item:full-anon-share} in \Cref{defi:full-anonymity}) holds. Note that the domain of the shares is $\F_q$.
        Fix $\bfv\in \Fq ^{\ell}$ for $\ell\leq k-1$ and distinct-element sequence $I \in [n]^{\ell}$.
        It holds that 
        \[\Pr_{\mathrm{sh}\leftarrow \mathtt{Share}(s)}[\text{sh}_I = \bfv] = \Pr_{\substack{f\leftarrow \Fq [x]\\ \deg(f)< k, f(0) = s}}[(f(\alpha_{I_1}),\ldots f(\alpha_{I_{\ell}})) = \bfv] = \frac{q^{k-\ell-1}}{q^{k-1}}=q^{-\ell}\;.
        \]
    \end{proof}

    \begin{remark} \label{rem:gap}
        The threshold $2k-1$ in \Cref{cor:ano-sec-ramp} is optimal for reconstruction procedures that recover the entire sharing polynomial from unlabeled evaluations. Indeed, suppose that $r<2k-1$ unlabeled evaluations always determined the entire polynomial. By \Cref{clm:correcting-cond}, the corresponding $[n,k]$ RS code would be robust against the $(n-r)$-permutation-insdel adversary and, in particular, would correct $n-r$ deletions. Since $n-r>n-2k+1\geq0$, this contradicts the half-Singleton bound \cite{abdel2007linear,cheng2020efficient}. This argument concerns recovery of the entire polynomial, not just its constant coefficient. It therefore does not establish optimality of the threshold for recovering only the secret $f(0)$.
    \end{remark}

    \fix{
    \begin{remark}
        A natural question is what happens inside the gap, namely, what information may be leaked from a multiset of shares of size $k-1 < \ell < 2k-1$. 
        For example, if the multiset is of the form $\{a,a,\ldots,a\}$, then the secret must be $a$, since the only polynomial of degree $<k$ that evaluates to $a$ in at least $k$ positions is the constant polynomial $f(x)\equiv a$. In this case, the set leaked the secret itself.
        More generally, to characterize the information leakage, one must determine all polynomials for which the permutation-insdel adversary can generate the given multiset of shares, and then analyze the distribution of their free coefficients. 
        This leads naturally to the list-decoding regime of Reed--Solomon codes against the permutation-insdel adversary, which we leave as an interesting direction for future research.
    \end{remark}
    }
    \paragraph{Explicit schemes.}
    \Cref{thm:main} and \Cref{cor:ano-sec-ramp} are both existential results. Namely, we do not specify the evaluation points for our RS codes or secret-sharing schemes.
    We show that the constructions presented in \cite{con2023optimal} and \cite{con2023reed} for RS codes correcting insdel codes are, in fact, RS codes that are robust against the permutation-insdel adversary. 

    The first scheme is borrowed from \cite{con2023optimal}. With slight modifications in the proof of \cite{con2023optimal}, \fix{we show that the two-dimensional RS code in \cite{con2023optimal} is robust against the $(n-3)$-permutation-insdel adversary},  thereby yielding a fully anonymous ramp secret-sharing scheme.
    \begin{cor}[see \Cref{thm:cnst-k-2}] \label{cor:cnst-k-2}
        Let $n \geq 3$. There is an explicit fully anonymous $(1,3,n)$ ramp secret-sharing scheme where the secret and the shares are elements of $\Fq$ for $q = O(n^3)$.
    \end{cor}

    The second scheme uses the evaluation points given in \cite[Construction 27]{con2023reed}. We prove the following result.
    \begin{cor}[see \Cref{thm:cnst-k-gen}] \label{cor:cnst-k-gen}
        Let $k$ and $n$ be integers such that $2k - 1 \leq n$. 
        There exists a deterministic construction of a fully anonymous $(k-1, 2k-1, n)$ ramp secret-sharing scheme where the secret and the shares are elements of $\Fq$ for $q=n^{O(k^2 ((2k)!))^2}$.
    \end{cor}
    We remark here that only for $k = O(\log n/\log \log n)$, the number of bits required to represent a share (or the secret) is polynomial in $n$. For higher values of $k$, it is super-polynomial.

    \paragraph{Comparison with \cite{eldridge2024abuse}}
    We highlight several key differences between the scheme presented in this paper and that of~\cite{eldridge2024abuse}.
    \begin{enumerate}
        \item \emph{Anonymity condition.} 
        Our ramp scheme requires that, for every fixed secret, the shares of every set of at most the privacy threshold $k-1$ are uniformly random and independent (condition~\ref{item:full-anon-share} in \Cref{defi:full-anonymity}). We impose no such requirement inside the ramp gap. As discussed above, the scheme of \cite{eldridge2024abuse} achieves a weaker notion of anonymity.
        \item \emph{Perfect vs. non-perfect reconstruction}. In~\cite{eldridge2024abuse}, reconstruction may fail with positive probability, whereas in our scheme, reconstruction always succeeds (information theoretically).
        \item \emph{Threshold vs. ramp scheme}. In \cite{eldridge2024abuse}, the scheme is 
        a threshold scheme. That is, any set of $k$ participants can recover the secret, while any set of $k - 1$ participants learns nothing about the secret.
        In our case, there is a factor-of-two gap between the number of participants who learn nothing about the secret ($k - 1$) and those who can reconstruct it ($2k - 1$).
        \item \emph{Size of a share.} In \cite{eldridge2024abuse}, to achieve a negligible reconstruction error, the size of a share is $\omega(\log(n))$, regardless of the reconstruction threshold. 
        In our case, the size of a share is $O(k\cdot \log(n))$ and so for constant $k$, our share size is asymptotically optimal up to constants~\cite{bogdanov2020threshold}.
        \item \emph{Explicitness}. Our main result, \Cref{cor:ano-sec-ramp}, is an existence result, as we do not specify the evaluation points that define the scheme. 
        However, in \Cref{cor:cnst-k-2} we provide an explicit fully anonymous $(1,3,n)$ ramp scheme over a field of size $O(n^3)$ and in \Cref{cor:cnst-k-gen} we provide a deterministic construction of a fully anonymous $(k-1, 2k-1, n)$ ramp scheme over a field of size $\approx n^{k^k}$.
        In \cite{eldridge2024abuse}, the authors provide a randomized scheme in the sense that the evaluation points are chosen uniformly at random from the field. 
        \item \emph{Reconstruction algorithm}. We do not provide an efficient reconstruction algorithm beyond brute-force search.
        In~\cite{eldridge2024abuse}, the scheme is both explicit and equipped with an efficient reconstruction algorithm.
    \end{enumerate}

    \subsection{Summary and future directions} \label{sec:open-questions}
    The notion of fully anonymous secret-sharing schemes was recently studied in~\cite{eldridge2024abuse,bishop2025fully}. One open question arising from these works—explicitly stated in~\cite{bishop2025fully}—is the construction of threshold schemes that are fully anonymous and offer perfect reconstruction. 
    In this paper, we show how RS codes that are robust against an adversary that first permutes the codeword and then deletes symbols give rise to a fully anonymous ramp secret-sharing scheme with perfect reconstruction.

    Many open questions stem from the comparison made earlier between our work and that of~\cite{eldridge2024abuse}.
    First, can the gap between the reconstruction and anonymity thresholds be reduced, or even closed, while still ensuring perfect reconstruction?
    As explained in \Cref{rem:gap}, reducing the reconstruction threshold in this approach would require recovering the secret without recovering the entire sharing polynomial. The half-Singleton argument alone does not rule out such a secret-only reconstruction procedure.
    
    Second, a clear drawback of our scheme is the absence of an efficient reconstruction algorithm beyond brute-force search.
    Providing such an algorithm is a tantalizing open question. 
    
    Third, the field size in our existential results implies a share size of $O(k \cdot \log n)$ bits.
    This is optimal only when $k$ is constant relative to $n$, and reducing the alphabet size of the underlying scheme remains an open challenge. 

    Another interesting direction is to study, more generally, which codes that are robust against deletions and permutations, such as those constructed in~\cite{kovavcevic2018codes}, can be transformed into fully anonymous secret-sharing schemes.

\subsection{Subsequent work and a revision}
After this work appeared in \emph{IEEE Transactions on Information Theory}~\cite{con2025anonymous}, Zhang, Sun, Zhang, and Ge~\cite{zhang2026random} made progress on the third open question in \Cref{sec:open-questions}, concerning the alphabet size.
Specifically, for fixed constants $R,\varepsilon\in(0,1)$ satisfying $2R+\varepsilon<1$ and $k=Rn$, they proved the existence of $[n,k]_q$ RS codes that are robust against the $((1-\varepsilon)n-2k+1)$-permutation-insdel adversary over fields of size $q=n^{O_{R,\varepsilon}(1)}$.
Thus, allowing an additive $\varepsilon n$ gap from the half-Singleton bound reduces the required alphabet size to polynomial in $n$, even at constant rate.

In~\cite[Remark~II.18]{zhang2026random}, Zhang et al. also pointed out an error in the auxiliary matrix used in the proof of \Cref{prop:determinant-implies-robustness} in the previous version of this paper.
In the present version, we give a simpler proof that avoids this matrix altogether.
This correction does not affect the results of the paper or their stated parameters; it changes only the proof.

The author used ChatGPT 6 Astra to assist with reviewing the manuscript for inconsistencies and improving its exposition.

\section{An algebraic condition}
This section is devoted to formulating an algebraic condition on the evaluation points that will ensure that the respective RS code is robust against $(n-2k+1)$-permutation-insdel adversary.
We start with recalling the main ideas and notations from \cite{con2023reed}.

\subsection{The framework of \cite{con2023reed}}	
	In this section, we recall the algebraic condition introduced in~\cite{con2023reed}.

    \begin{defi}[$V$-matrix]\label{matrix} 
        Let $X_1,\ldots,X_n$ be formal variables. For positive integers $k$ and two distinct-element sequences $I,J \in[n]^{2k-1}$, define the $(2k-1)\times(2k-1)$ matrix
        \begin{equation} \label{eq:V-matrix}
            V_{I,J}(X_1,X_2,\ldots,X_n):=\left(\begin{array}{ccccccc}1 & X_{I_1} & \cdots & X_{I_1}^{k-1} & X_{J_1} & \cdots & X_{J_1}^{k-1} \\1 & X_{I_2} & \cdots & X_{I_2}^{k-1} & X_{J_2} & \cdots & X_{J_2}^{k-1} \\ \vdots& \vdots & \ddots& \vdots & \vdots & \ddots & \vdots \\1 & X_{I_{2k-1}} & \cdots & X_{I_{2k-1}}^{k-1} & X_{J_{2k-1}} & \cdots & X_{J_{2k-1}}^{k-1}\end{array}\right),
        \end{equation}
    \end{defi}

    We say that a distinct-element sequence $I$ is an \emph{increasing distinct-element sequence} if $I_1 < \cdots < I_{|I|}$. We say that two sequences $I,J$ agree on a coordinate $\ell$, if $I_{\ell} = J_{\ell}$.
     The following algebraic condition was given in \cite{con2023reed}.

    \begin{prop}\cite[Proposition 14]{con2023reed}\label{bad}
	Consider the code $\RS$. If, for every    pair of increasing distinct-element sequences $I, J\in [n]^{2k-1}$ that agree on at most $k-1$ coordinates, it holds that $\det(V_{I, J}(\ba)) \neq 0$, then the code can correct any $n-2k+1$ insdel errors.

		Moreover, if the code can correct  any $n-2k+1$ insdel errors, then the only vectors in the right kernel of  $V_{I, J}(\ba)$ are of the form $(0,f_1,\ldots,f_{k-1},-f_1,\ldots,-f_{k-1})$.
\end{prop}

    Recall that $\bX$ and $\ba$ refer to $(X_1, \ldots, X_n)$ and $(\alpha_1, \ldots, \alpha_n)$, respectively.
    A key technical lemma concerning the $V$-matrices, proved in~\cite{con2023reed}, is as follows.
\begin{prop}\cite[Proposition 18]{con2023reed}\label{nonzero}
    Let $I, J \in[n]^{2k-1}$ be two increasing distinct-element sequences that agree on at most $k-1$ coordinates.  
    Then, we have 
    \[
    \det\left(V_{I,J}(\bX)\right)\neq 0
    \] 
    as a multivariate polynomial in $\F_q[X_1,X_2,\ldots,X_n]$.
\end{prop}
    Informally speaking, using this proposition and considering all relevant $V$-matrices, the authors of~\cite{con2023reed} applied the Schwarz–Zippel–DeMillo–Lipton lemma~\cite{schwartz1980fast,zippel1979probabilistic,demillo1978probabilistic} \fix{to show} that there exists an assignment $\ba=(\alpha_1, \ldots, \alpha_n)$ over a sufficiently large field for which $\det (V_{I, J}(\ba))$ is nonzero for all pairs of increasing distinct-element sequences $I,J$ of length $2k-1$ that agree on at most $k-1$ coordinates.

    \subsection{A sufficient and necessary condition}
    One might ask whether the proof in~\cite{con2023reed} for the existence of $[n,k]$ RS codes that can correct $n - 2k +1$ insdel also implies the existence of $[n,k]$ RS codes that can correct an arbitrary permutation followed by $n - 2k + 1$ insdel errors, i.e., that are robust against the $(n - 2k + 1)$-permutation-insdel adversary. As shown in the following remark, the proof in~\cite{con2023reed} fails when permutations are also considered.

    \begin{remark}
         We note that \Cref{nonzero} is not true if we remove the requirement that $I,J$ are increasing distinct-element sequences. Indeed, let $k = 4$, $\ell = 7$, and define $I = (1,2,3,4,5,6,7)$ and $J = (2,1,4,3,5,6,7)$. Note that they agree on the three last coordinates (thus, they agree on $\leq k-1$ coordinates). In this case, $V_{I,J}$ is
         \[
         \left(\begin{array}{ccccccc}1 & X_{1} &  X_{1}^2 & X_{1}^{3} & X_{2} & X_{2}^2 & X_{2}^{3} \\
         1 & X_{2} &  X_{2}^2 & X_{2}^{3} & X_{1} & X_{1}^2 & X_{1}^{3} \\
         1 & X_{3} &  X_{3}^2 & X_{3}^{3} & X_{4} & X_{4}^2 & X_{4}^{3} \\
         1 & X_{4} &  X_{4}^2 & X_{4}^{3} & X_{3} & X_{3}^2 & X_{3}^{3} \\
         1 & X_{5} &  X_{5}^2 & X_{5}^{3} & X_{5} & X_{5}^2 & X_{5}^{3} \\
         1 & X_{6} &  X_{6}^2 & X_{6}^{3} & X_{6} & X_{6}^2 & X_{6}^{3} \\
         1 & X_{7} &  X_{7}^2 & X_{7}^{3} & X_{7} & X_{7}^2 & X_{7}^{3} \\
         \end{array}\right),
        \]
        and one can verify by basic column and row operations that the determinant is zero. Thus, in order to show robustness against the $(n-2k+1)$-permutation-insdel adversary, we also need to study singular $V_{I,J}(\bX)$ since we have to support pairs of distinct-element sequences for which $V_{I,J}(\bX)$ is singular. 
    \end{remark}

    We will have to make sure that our evaluation points $\ba$ are such that the right kernel of $V_{I,J}(\ba)$ is trivial or contains only elements of the form $(0,f_1,\ldots,f_{k-1},-f_1,\ldots,-f_{k-1})$.
    Formally, we have the following necessary and sufficient condition.

    \begin{prop} \label{prop:algebraic-condition}
    Let $\alpha_1, \ldots, \alpha_n\in \Fq$ be distinct points. The following are equivalent:
    \begin{enumerate}
        \item For every two distinct-element sequences $I,J\in [n]^{2k-1}$ it holds that either $\det(V_{I,J} (\ba)) \neq 0$ or any vector in the kernel of $V_{I,J} (\ba)$ is of the form
        \[(0, f_1, \ldots, f_{k-1}, -f_1, \ldots, -f_{k-1}) \;.\] 
        \item The code $\RS$ is robust against the $(n-2k+1)$-permutation-insdel adversary. 
    \end{enumerate}
\end{prop}

\begin{proof}
    $(1) \rightarrow (2)$. We prove the contrapositive, i.e., $\neg(2) \rightarrow \neg (1)$. Assume that the code $\mathsf{RS}_{n,k}(\alpha_1, \ldots, \alpha_n)\subseteq\F_q^n$ is not robust  against the $(n-2k+1)$-permutation-insdel adversary. Then, by \Cref{clm:correcting-cond}, there exist two distinct codewords $c$ and $c'$ that correspond to two distinct polynomials $f=\sum_if_ix^i$ and $g=\sum_ig_ix^i$, respectively, where $\deg(f),\deg(g)<k$, and two distinct-element sequences $I, J\in [n]^{2k-1}$ such that 
    \begin{equation}\label{f1f2}
        f(\alpha_{I_s})=g(\alpha_{J_s}) \quad \text{for } s=1,2,\dots,2k-1 \;,
    \end{equation}

    Consider the vector 
    \[
    \bfv:=\left(f_0-g_0 ,f_1,f_2,\dots,f_{k-1},-g_1,-g_2,\dots,-g_{k-1}\right)\in\F_q^{2k-1} \;,
    \]
    which is non-zero and not of the form $(0, f_1, \ldots, f_{k-1}, -f_1, \ldots, -f_{k-1})$ but is inside the kernel of $V_{I,J}(\alpha)$. 

    $(2) \rightarrow (1)$. Again, we prove it in the contrapositive way. Assume there exist two distinct-element sequences $I, J\in [n]^{2k-1}$ such that the right kernel of $V_{I,J}(\ba)$ contains 
    \[
    \bfv = (f_0, f_1, \ldots, f_{k-1}, -g_1, \ldots, -g_{k-1})\;,
    \]
    where either $f_0 \neq 0$ or there exists $i\in [k-1]$ such that $f_i \neq g_i$. In particular, this implies that the polynomials $f = \sum_{i=0}^{k-1} f_i x^i$ and $g = \sum_{i=1}^{k-1} g_i x^i$ are distinct and that $f(\alpha_{I_s}) = g(\alpha_{J_s})$ for all $s\in [2k-1]$. Since both $f$ and $g$ are polynomials of degree $<k$, the corresponding two codewords $c$ and $c'$ are such that $c_I = c'_J$, which implies by \Cref{clm:correcting-cond} that $\RS$ is not robust against the $(n-2k+1)$-permutation-insdel adversary.
\end{proof}

\section{Algebraic structure and a refined condition on the evaluation points}

In this section, we state a condition for the evaluation points of the RS code.
Then, in the next section, we will show that if this condition holds, then the resulting RS code is robust against the $(n-2k+1)$-permutation-insdel adversary.
We begin by citing the well-known Schwarz-Zippel-Demillo-Lipton lemma. 
\begin{lemma} \cite{schwartz1980fast,zippel1979probabilistic,demillo1978probabilistic}
    Let $P(X_1, \ldots, X_n) \in \Fq[X_1, \ldots, X_n]$ be a nonzero polynomial of total degree $d\geq 0$. Let $S \subseteq \Fq$, then $\Pr_{\alpha_1, \ldots, \alpha_n \sim S} [P(\alpha_1, \ldots, \alpha_n)=0]\leq d/|S|$. \label{lem:sz-lemma}
\end{lemma}

\label{sec:first-cond}
We take a closer look at $\det(V_{I,J}(\bX))$ and identify the conditions under which it is nonzero. Then, we show that with high probability, if $\alpha_1, \ldots, \alpha_n$ are chosen uniformly at random from a sufficiently large field, it holds that $\det(V_{I,J}(\ba)) \neq 0$ whenever the conditions on $(I,J)$ are satisfied.

Note that we always write the elements of a set in increasing order. Thus, a set can be regarded as an increasing distinct-element sequence (and vice versa). Therefore, for brevity and to enable the use of standard set operations (such as inclusion and set difference), we will refer to these sequences simply as sets.

\begin{defi}
    Let $I\in[n]^s$ be a distinct-element sequence and let $A\subseteq [s]$ be a set. Then, $I_A$ denotes the distinct-element sequence $(I_{A_1}, I_{A_2}, \ldots, I_{A_{|A|}})$.    
\end{defi}

We now define the notion of a minimal equal sub-pair between two distinct-element sequences.
    \begin{defi}
        Let $(I,J)$ be a pair of distinct-element sequences of the same length, $s$. Let $A \subseteq [s]$ be a non-empty set of indices. We call the pair $(I_A, J_A)$ a \emph{minimal equal sub-pair} if
        \begin{itemize}
            \item $I_A = J_A$ as sets.
            \item For every nonempty proper subset $A' \subsetneq A$, $I_{A'} \neq J_{A'}$ as sets.
        \end{itemize}
    \end{defi}
    \begin{example}
        Let $I = (1,2,3,7,8)$ and $J = (2,3,1,8,7)$. Then for $A_1 = \{1,2,3\}$ we see that $(I_{A_1}, J_{A_1})$ is a minimal equal sub-pair. In addition, $A_2 = \{4,5\}$ induces a minimal sub-pair since as sets $I_{A_2} = J_{A_2} = \{7,8\}$.

        For $I = (6,7,8,9,5)$ and $J = (7,6,9,8,1)$, we note that $A = \{1,2,3,4\}$ is not a minimal equal sub-pair since $A' = \{1,2\} \subset A$ is such that 
        $I_{A'} = J_{A'}$.
    \end{example}

    We now have the following simple claim.
    \begin{claim} \label{clm:no-intersection-eq-sub-pair}
        Let $(I,J)$ be a pair of distinct-element sequences. Let $A_1, A_2$ be two distinct sets such that $(I_{A_1}, J_{A_1})$ and $(I_{A_2}, J_{A_2})$ are minimal equal sub-pairs. Then, $A_1 \cap A_2 = \emptyset$.
    \end{claim}
    \begin{proof}
        Assume for contradiction that \fix{$A_1\cap A_2\neq \emptyset$}. Since both $(I_{A_1}, J_{A_1})$ and $(I_{A_2}, J_{A_2})$ are minimal equal sub-pairs, it must be that for $A' = A_1 \cap A_2$ we have $I_{A'} \neq J_{A'}$ as sets; otherwise, it contradicts the minimality of $A_1$ and $A_2$. Thus, there is $a\in A'$ such that $I_a \notin J_{A'}$. Moreover, $I_a\in I_{A_1} = J_{A_1}$ and $I_a\in I_{A_2} = J_{A_2}$ as sets. Therefore, $I_a \in J_{A_1 \setminus A'}$ and $I_a \in J_{A_2 \setminus A'}$. But $J_{A_1 \setminus A'}$ and $J_{A_2 \setminus A'}$ do not share any common elements. We arrive at a contradiction.
    \end{proof}

    \fix{We shall call a pair of distinct-element sequences $(I,J)$ \emph{free of equalities} if it has no minimal equal sub-pairs. Formally,}
    \begin{defi}
        A pair of distinct-element sequences $(I,J)$ of length $s$, is said to be \emph{free of equalities} if for every nonempty set $A \subseteq [s]$, $I_A \neq J_A$ as sets.
    \end{defi}

    We next define a full decomposition of $(I,J)$ into minimal equal sub-pairs.
    \begin{defi} \label{def:I-J-decomposition}
        Let $(I,J)$ be a pair of distinct-element sequences of length $s$. We call a family of sets $A_1, \ldots, A_t, B \subseteq [s]$ a \emph{decomposition} of $(I,J)$ if:
        \begin{enumerate}
            \item For any two distinct $j,j' \in [t]$, we have that $A_j \cap A_{j'} = \emptyset$.
            \item $(I_{A_j}, J_{A_j})$ is a minimal equal sub-pair for every $j\in [t]$.
            \item $B = [s] \setminus \left( \sqcup_{j\in [t]} A_j\right)$ and $(I_B,J_B)$ is free of equalities.\footnote{$\sqcup$ refers to disjoint union}
        \end{enumerate}
    \end{defi}

    As a corollary, we get the following claim.
    \begin{claim}
        Let $(I,J)$ be a pair of distinct-element sequences of length $s$. Then, up to reordering of the sets, there exists a unique decomposition of $(I,J)$.
    \end{claim}
    \begin{proof}
        The claims about the existence and the uniqueness follow from the same observation: For each $i\in [s]$, there are only two options. 
        Either $i$ belongs to a \emph{single} $A_j$ for which $(I_{A_j}, J_{A_j})$ is a minimal equal sub-pair, or $i$ does not belong to such a set, and thus $i\in B$.
        Suppose we have two decompositions $A_1, \ldots, A_t, B$ and $A'_1, \ldots, A'_{t'}, B'$ of $(I,J)$. Then, by \Cref{clm:no-intersection-eq-sub-pair}, for each $i\in [t], j\in [t']$ we have either $A_i \cap A'_j = \emptyset$, or $A_i = A'_j$. Without loss of generality, assume that $t' \geq t$.
        If these are distinct decompositions, there exists $j\in [t']$ such that $A'_j \cap A_i = \emptyset$ for all $i\in [t]$, which implies that $A'_j\subset B$, in contradiction to  $(I_B,J_B)$ being free of equalities.
    \end{proof}
    We shall need the following trivial claim.
    \begin{claim}
    \label{clm:two-dis-ind-in-free}
        Let $s\geq 2$ be an integer and $(I,J)$ be a pair of distinct-element sequences of length $s$ that is free of equalities.
        Then there exist two distinct indices $i\neq j\in [s]$ such that $I_i\notin J$ and $J_j\notin I$. 
    \end{claim}
    \begin{proof}
        Clearly, $I\neq J$ as sets, so there exists $i\in [s]$ for which $I_i \notin J$. Now, since $(I,J)$ is free from equalities, $I_{[s]\setminus \{i\}} \neq J_{[s]\setminus \{i\}}$ as sets. Thus, there must exist $j\in [s]\setminus \{i\}$ such that $J_j \notin I$.
    \end{proof}

    In the following proposition, we show that if $I,J$ are distinct-element sequences with at most $k$ minimal equal sub-pairs, then the determinant of the corresponding $V$-matrix is not zero. Formally, 

    \begin{prop} \label{prop:formal-det-1}
        Let $I,J \in [n]^{2k-1}$ be two distinct-element sequences.
        Assume that $A_1, \ldots, A_t, B \subseteq [2k-1]$ is a decomposition of $(I,J)$ according to \Cref{def:I-J-decomposition}.
        If $t\leq k$, then in the expansion of $\det(V_{I,J}(\bX))$ as a sum over permutations, there is a monomial that is obtained at exactly one of the $(2k-1)!$ different permutations. 
        In particular, its coefficient is $\pm 1$, depending on the sign of the corresponding permutation. Consequently, $\det(V_{I,J}(\bX))\neq 0$. 
    \end{prop}
    \begin{proof}
        We prove the proposition by induction on $k$. For $k=1$, we have $V_{I,J}(\bX)=1$, and the result follows immediately. For the induction step, assume it holds for $k-1$, and we prove it for $k\geq 2$. 
		Consider two coordinates $i,j$, determined according to the following three cases.
        \begin{enumerate}
            \item $t > 0$ and there exists $s\in [t]$ such that $|A_s| \geq 2$. 
            In this case, pick $i,j\in A_s$ to be such that $i\neq j$ and $I_i = J_j$. \label{prop:case1}
            \item $t > 0$ and for all $s\in [t]$, $|A_s| = 1$. In this case, let $i$ to be an element of one of the sets $A_s$. 
            Now, let $j$ be such that $j\in B$ and $J_j\notin I$. Clearly there exists such since ($I_{B}, J_{B}$) is free from equalities. \label{prop:case2}
            \item $t = 0$. In this $(I,J)$ is free from equalities, so choose $i\neq j$ according to \Cref{clm:two-dis-ind-in-free}.\label{prop-case3}
        \end{enumerate}
        We shall analyze case~\ref{prop:case1} first. In the determinant expansion of $V_{I,J}$ as a sum of $(2k-1)!$ monomials, collect all the monomials that are divisible by $X_{I_i}^{k-1}X_{J_j}^{k-1} = X_{I_i}^{2k-2}$, and group them as
	\[
            X_{I_i}^{2k-2}f(\bX)\;,
        \]
		for some polynomial $f$ in the variables $(X_\ell : \ell \in (I\cup J) \setminus\{I_i\}))$. 
        Any monomial in  the determinant expansion of $V_{I,J}$ that is divisible by $X_{I_i}^{2k-2}$ must be obtained by picking the $(i,k)$ and the $(j,2k-1)$ entries in the matrix $V_{I,J}(\bX)$ (see \eqref{eq:V-matrix}).

        Now, define $I' = I_{[2k-1] \setminus \{i,j\}}$ and $J' = J_{[2k-1] \setminus \{i,j\}}$ and observe that the $(2k-3)\times(2k-3)$ matrix $V_{I', J'}$ is obtained from $V_{I,J}$ by removing the rows $i,j$ and columns $k, 2k-1$ and therefore $f(\bX) = \pm\det(V_{I',J'}(\bX))$.
        We now claim that $I'$ and $J'$ comply with the conditions stated in the proposition.
        First note that $A' = A_s \setminus \{i,j\}$ is a set for which $(I_{A'}, J_{A'})$ is no longer a \fix{minimal} equal sub-pair (in fact, it is free of equalities).
        We also need to show that the elements of $A'$ together with (some) elements $B$ do not contain a new set that gives rise to an equally minimal sub-pair. 
        Indeed, if there is a $A'' \subset A'\sqcup B$ for which $(I_{A''}, J_{A''})$ is a minimal equal sub-pair, then it must be that there are distinct $b_1,b_2 \in A''$ such that $b_1 \in A'$ and $b_2\in B$. But note that in this case, $A''$ and $A_s$ have a nonempty intersection and both correspond to minimal equal sub-pair of $(I,J)$, which contradicts \Cref{clm:no-intersection-eq-sub-pair}.
        Thus, the number of sets that correspond to minimal equal sub-pairs in $(I', J')$ is $t-1 \leq k-1$. 
        Hence, by the induction hypothesis, $\det(V_{I',J'})$ contains a nonzero monomial $m$ (with a $\pm 1$ coefficient) that is uniquely obtained among the $(2k-3)!$ different monomials. Hence, $X_{I_i}^{2k-2}m$ is a monomial of $X_{I_i}^{2k-2}f(\bX)$ with a $\pm 1$ coefficient. 
        Since this monomial cannot arise from any other term in the determinant expansion of $V_{I,J}(\bX)$, this monomial is uniquely obtained in $\det(V_{I,J}(\bX))$, and the result follows. 

        We now analyze cases~\ref{prop:case2} and~\ref{prop-case3}. The argument here will be similar to the one of case~\ref{prop:case1} and actually almost identical to \cite[Proposition 18]{con2023reed}. We include it here for completeness.

        In the determinant expansion of $V_{I,J}$ as a sum of $(2k-1)!$ monomials, collect all the monomials that are divisible by $X_{I_i}^{k-1}X_{J_j}^{k-1}$ (here, observe that $I_i \neq J_j$), and write them together as
        \[
        X_{I_i}^{k-1}X_{J_j}^{k-1}f(\bX)\;,
        \] 
        for some polynomial $f$ in the variables $(X_\ell : \ell \in (I\setminus\{I_i\}) \cup (J\setminus\{J_j\}))$. 
        By the choice of $i$ and $j$ in these two cases, any monomial in the determinant expansion of $V_{I,J}$ that is divisible by $X_{I_i}^{k-1}X_{J_j}^{k-1}$ must be obtained by picking the $(i,k)$ and the $(j,2k-1)$ entries in the matrix $V_{I,J}(\bX)$ (see \eqref{eq:V-matrix}). 
        Define $I' = I_{[2k-1] \setminus \{i,j\}}$ and $J' = J_{[2k-1] \setminus \{i,j\}}$ and observe that the $(2k-3)\times(2k-3)$ matrix $V_{I', J'}$ is obtained from $V_{I,J}$ by removing the rows $i,j$ and columns $k, 2k-1$ and as a result $f(\bX) = \pm\det(V_{I',J'}(\bX))$.
        We now claim that $I', J'$ satisfy the conditions of the proposition. Indeed, we have reduced the $t$ by $1$ if $t>0$ and thus the number of sets that give rise to \fix{minimal} equal sub-pairs is $\leq k-1$. Also, observe that if a pair of distinct-element sequences $(I,J)$ is free from equalities, then $(I_{[2k-1]\setminus\{i\}}, J_{[2k-1]\setminus\{i\}})$ is also free from equalities. This concludes that $I',J'$ comply with the conditions for $k-1$. Thus, by the induction hypothesis, $\det(V_{I',J'})$ contains a nonzero monomial $m$ (with a $\pm 1$ coefficient) that is uniquely obtained among the $(2k-3)!$ different monomials. 
        Hence,  $X_{I_i}^{k-1}X_{J_j}^{k-1}m$ is a monomial of $X_{I_i}^{k-1}X_{J_j}^{k-1}f$ with a $\pm 1$ coefficient. Since there is no other way to obtain this monomial in the determinant expansion of $V_{I,J}(\bX)$, this monomial is uniquely obtained in $\det(V_{I,J}(\bX))$, and the result follows for these cases as well.
    \end{proof}

    Our next proposition shows that over a large enough field, a random selection of $n$ points $\alpha_1,\ldots, \alpha_n$ from that field will satisfy the following two conditions: (i) The points are distinct and (ii) For every distinct-element sequences $I,J\in [n]^{2k-1}$ containing at most $k$ minimal equal sub-pairs, it holds that $\det(V_{I,J}(\ba))\neq 0$. Specifically,
    \begin{prop}
        \label{prop:existance-first-cond}
        Let $k\geq 2$ and $n\geq 2k-1$ be integers, and let $q$ be a prime such that
        \[
        q\geq 200 \cdot \binom{n}{2k - 1} ^2 \cdot (2k-1)! \cdot k(k-1)\;.
        \] 
        Let $\ba = (\alpha_1, \ldots, \alpha_n) \in \Fq^n$ be a uniformly random vector.
        Then, with probability at least $0.99$, 
        for every pair of distinct-element sequences of size $2k-1$ containing at most $k$ minimal equal sub-pairs, it holds that all the points in $\ba$ are pairwise distinct and $\det(V_{I,J}(\ba)) \neq 0$.
    \end{prop}
    \begin{proof}
        Let $I=(I_1, \ldots, I_{2k-1})$ and $J = (J_1, \ldots, J_{2k-1})$ be two distinct-element sequences, and let $\pi$ be a permutation on $[2k-1]$ and define $\pi (I) = (I_{\pi(1)}, \ldots, I_{\pi(2k-1)})$ and $\pi(J) = (J_{\pi(1)}, \ldots, J_{\pi(2k-1)})$. Then, $\det (V_{I,J}) = \pm \det(V_{\pi(I), \pi(J)})$. Thus, two pairs of distinct-element sequences $(I,J)$ and $(I', J')$ will be equivalent, $(I,J)\sim (I',J')$, if there exists $\pi$ such that $I = \pi(I')$ and $J = \pi(J')$.
        Define 
        \[
        F(\bX)=\prod_{i< j}(X_i-X_j)\prod_{(I,J)\in S/\sim}\det(V_{I,J}(\bX))\;,
        \]
        where $S$ is the set containing all the pairs of distinct-element sequences that contain at most $k$ \fix{minimal} equal sub-pairs and so the second product runs over all equivalence classes under $\sim$. 
        By \Cref{prop:formal-det-1}, every determinant factor remains nonzero after reduction modulo any prime, since it has a coefficient equal to $\pm1$. The Vandermonde factors are also nonzero over every field. 
        Hence $F(\bX)$ is a nonzero polynomial in $\F_q[\bX]$,  and its degree can be upper bounded by
        \begin{align*}
            \deg(F) &\leq \binom{n}{2} + \binom{n}{2k - 1} ^2 \cdot (2k-1)! \cdot k(k-1) \\
            &\leq 2\cdot \binom{n}{2k - 1} ^2 \cdot (2k-1)! \cdot k(k-1)
        \end{align*}
        Indeed, the number of pairs of distinct-element sequences of size $2k-1$ that contain at most $k$ \fix{minimal} equal sub-pairs is at most $\binom{n}{2k-1}^2\cdot ((2k-1)!)^2$ but since we consider only equivalence classes, we divide by $(2k-1)!$.
        Furthermore, one can easily see that the total degree of each $\det(V_{I,J}(\bX))$ is $k(k-1)$. The $\binom{n}{2}$ term corresponds to the degree of $\prod_{i< j}(X_i-X_j)$.

        Therefore, by \Cref{lem:sz-lemma} and the assumed lower bound on $q$, with probability at least $0.99$, a uniformly random $\ba = (\alpha_1, \ldots, \alpha_n)$ is such that $F(\ba) \neq 0$ and clearly, by the structure of $F(\bX)$, we have $\alpha_i \neq \alpha_j$ for all distinct $i,j\in [n]$.
        The claim follows by noting that for any $(I,J)\in S/\sim$, it holds that $\det( V_{I,J}(\ba) ) \neq 0$ and by the definition of the equivalence relation, $\det (V_{I,J}(\ba) )\neq 0$ for any $(I,J)\in S$.
    \end{proof}

     \section{Existence of $[n,k]$ RS codes robust to $(n-2k+1)$-permutation-insdel}
    We now prove the main theorem of this paper. We first show that the determinant condition from the previous section implies the algebraic condition in \Cref{prop:algebraic-condition}.

     \begin{prop}  \DIFaddbegin \label{prop:determinant-implies-robustness}
        Let $k,n$ be positive integers with $2k-1\leq n$, and let $\ba = (\alpha_1, \ldots, \alpha_n) \in \Fq^n$ be a vector of pairwise distinct points. 
        Assume that, for every pair $(I,J)$ of distinct-element sequences of length $2k-1$ that contains at most $k$ minimal equal sub-pairs, it holds that $\det(V_{I,J}(\ba)) \neq 0$.
        Then, the respective $\RS$ code is robust against the $(n-2k+1)$-permutation-insdel adversary.
    \end{prop}
    \begin{proof}
        According to \Cref{prop:algebraic-condition}, in order for the $\RS$ to be robust against the $(n-2k+1)$-permutation-insdel adversary, we need to show that for every pair of distinct-element sequences $I,J$, of length $2k-1$, the only elements in the right kernel (if any exist) of $V_{I,J}(\ba)$ are of the form
        \begin{equation} \label{eq:kernel-des-form}
         (0, f_1, \ldots, f_{k-1}, -f_1, \ldots, -f_{k-1}) \;.    
        \end{equation}
        Fix two distinct-element sequences $I,J\in[n]^{2k-1}$. Associate with $(I,J)$ the directed graph on vertex set $I\cup J$ whose edge indexed by $s\in[2k-1]$ is $J_s\to I_s$. Note that this graph can have self loops.
        Since both sequences have distinct elements, every vertex has in-degree and out-degree at most one. For a nonempty set $A\subseteq[2k-1]$, we have $I_A=J_A$ as sets if and only if, in the directed subgraph formed by the edges indexed by $A$, every vertex has equal in-degree and out-degree. 
        Such a set of edges is a union of directed cycles, and it is minimal precisely when it is one directed cycle (a self loop is regarded as a cycle). 
        Consequently, the minimal equal sub-pairs of $(I,J)$ are precisely the directed cycles of this graph.

        Suppose, towards a contradiction, that the right kernel of $V_{I,J}(\ba)$ contains a vector that is not of the form~\eqref{eq:kernel-des-form} and write it as
        \begin{equation} \label{eq:contradiction-vector}    
        (f_0,f_1,\ldots,f_{k-1},-g_1,\ldots,-g_{k-1})\;.
        \end{equation}
        Set $f(x)=\sum_{i=0}^{k-1}f_ix^i$, $g(x)=\sum_{i=1}^{k-1}g_ix^i$, and observe that
        \begin{equation} \label{eq:original-collision}
            \textup{for all } s\in [2k-1]:\;\;f(\alpha_{I_s})=g(\alpha_{J_s}) \;.
        \end{equation}
        Define $h:=f-g$ and note that $h\not \equiv 0$ since otherwise we would have $f_0 = 0$ and $f_i = g_i$ for all $i\in [k-1]$ which contradicts the assumption that the vector~\eqref{eq:contradiction-vector} is not of the form~\eqref{eq:kernel-des-form}.

        As mentioned above, we are interested in directed cycles $B\subseteq [2k-1]$ (that give rise to minimal equal sub-pair $(I_B, J_B)$) for which we have
        \[
        \sum_{s\in B} h(\alpha_{I_s}) = \sum_{s\in B}(f(\alpha_{I_s}) - g(\alpha_{I_s})) = 0 = \sum_{s\in B} h(\alpha_{J_s}) \;.
        \]
        
         Let $t$ be the number of  directed cycles. If $t\leq k$, then the condition in the proposition states that $\det(V_{I,J}(\ba))\neq0$,  contradicting the existence of the kernel vector. We may therefore assume that $t>k$. 

         We next modify $J$ while keeping $I$ fixed. Choose edges $J_a\to I_a$ and $J_b\to I_b$ from two distinct directed cycles.
         Then, interchange the values $J_a$ and $J_b$ in $J$. This clearly leaves $J$ a distinct-element sequence. 
         Note that this operation affects the graph as follows. These two edges, i.e., $J_a\to I_a$ and $J_b\to I_b$ are deleted, and two new edges: $J_a \to I_b$ and $J_b \to I_a$ are created.
         This operation merges the two distinct directed cycles into a single directed cycle, and leaves all other cycles and path components unchanged. 
         
        Repeating this operation $t-k$ times gives a distinct-element sequence $J'$ for which the graph associated with $(I,J')$ has exactly $k$ directed cycles, whose sets of edge indices are denoted by $C_1,\ldots,C_k \subseteq [2k-1]$ (recall that we denote an edge with its index in the $2k-1$ edges $J'_a \to I_a$ for $a\in [2k-1]$). 
        Now, for every cycle $C\in \{C_1, \ldots, C_k\}$, we have that
        \begin{equation} \label{eq:h-cycle-sum}
            \sum_{s\in C}h(\alpha_{I_s})=\sum_{s\in C}h(\alpha_{J'_s}) = 0\;.
        \end{equation}
        Indeed, the vertex set of every $C \in \{C_1, \ldots, C_k\}$ is a disjoint union of vertex sets of directed cycles in the original graph $(I,J)$, and furthermore, the new cycle $C$ is a directed cycle in the graph induced by $(I, J')$ and thus, $I_C = J'_C$ as sets. 
        
        Define the linear map
        \begin{align*}
            D:\F_q^{<k}[x]&\longrightarrow\F_q^{2k-1}\\
            p& \longmapsto\bigl(p(\alpha_{I_s})-p(\alpha_{J'_s})\bigr)_{s\in[2k-1]}\;,
        \end{align*}
        and observe that $\ker D$ consists exactly of the constant polynomials. 
        Indeed, if $p(x)=\sum_{i=0}^{k-1}p_ix^i$ is a nonconstant polynomial in $\ker(D)$, then the following nonzero vector, which corresponds to $p - p_0$,
        \[
            (0,p_1,\ldots,p_{k-1},-p_1,\ldots,-p_{k-1})\in\ker(V_{I,J'}(\ba)) \;,
        \]
        in contradiction to $\det(V_{I,J'}(\ba)) \neq 0$. Thus, we have that $\textup{rank}(D) = k-1$.
        Also observe that for all $p(x)$ of degree $< k$, we have that 
        \[
        \forall r\in [k] : \;\; \sum_{s\in C_r} D(p)_s = 0 \;,
        \]

        and for
        \[
            W:=\left\{z\in\F_q^{2k-1}:\sum_{s\in C_r}z_s=0\text{ for every }r\in[k]\right\},
        \]
        it holds $\textup{Im}(D) \subseteq W$. Further, $\dim (W) = k-1$ since there is one coordinate in each cycle that is linearly dependent on the other coordinates in the cycle and all cycles are disjoint. Consequently, $\textup{Im}(D) = W$.
        
        Therefore, since $ (h(\alpha_{J'_s}))_{s\in [2k-1]} \in W$, there exists a polynomial $f'(x) \in \F_q[x]_{<k}$ such that 
        \[
        D(f') = \left(f'(\alpha_{I_s}) - f'(\alpha_{J'_s}) \right)_{s\in [2k-1]} = \left(h(\alpha_{J'_s})\right)_{s\in [2k-1]} \;.
        \]
        Set $g'=f'+h$. For all $s\in [2k-1]$, we have
        \[
        g'(\alpha_{J'_s}) = f'(\alpha_{J'_s}) + (f'(\alpha_{I_s}) - f'(\alpha_{J'_s})) = f'(\alpha_{I_s}) \;.
        \]

        Now, since $h$ is nonzero, we have that $f' \neq g'$ and thus, the vector
        \[
        (f'_0 - g'_0, f'_1, \ldots, f'_{k-1}, -g'_1, \ldots, -g'_{k-1})
        \]
        is nonzero and inside the kernel of $V_{I,J'}(\ba)$. In contradication to $\det(V_{I,J'}(\ba)) \neq 0$.
        \end{proof}
        
        We are now ready to prove our main theorem, which is restated for convenience.

    \mainres*
    \begin{proof}
        Given \Cref{prop:determinant-implies-robustness}, it suffices to find a vector $\ba = (\alpha_1, \ldots, \alpha_n)$  satisfying the determinant condition in that proposition.

        The prime $q$ in the statement satisfies
        \[
            q\geq \OurfieldSize\geq 200\binom{n}{2k-1}^2(2k-1)!k(k-1).
        \]
         By \Cref{prop:existance-first-cond}, a uniformly random  vector $\ba\in\F_q^n$ satisfies the required condition with probability at least $0.99$. We conclude that such a vector exists.
    \end{proof}

\section{Explicit constructions}
In this section, we will show that the previous constructions of RS codes that can correct insdel errors \cite{con2023reed,con2023optimal} are, in fact, constructions of RS codes that are robust against the permutation-insdel adversary.

\subsection{The two-dimensional case}
\fix{In this section, we use the second construction from \cite{con2023optimal}, which assumes an underlying field of characteristic $>2$. 
The proof, however, works also with the first construction in \cite{con2023optimal}, which is defined over fields of characteristic $2$.}
We start by recalling the construction
\begin{cnst}[Proposition 2.3 \cite{con2023optimal}] \label{const:explicit-n-2}
    Let $\Fq$ be a finite field of characteristic $p>2$ and let $\mathcal{A} \subseteq \Fq^*$ be a subset of size $n$. Let $\gamma$ be a root of a degree $3$ irreducible polynomial over $\Fq$, and set $K:=\Fq(\gamma)\cong\mathbb{F}_{q^3}$.
    Define the $\text{RS}_{n,2}(\alpha_1, \ldots, \alpha_n)$ code over $K$, where the points $\alpha_1, \ldots, \alpha_n$ are some ordering of $\delta + \delta^2 \cdot \gamma$, for $\delta \in \mathcal{A}$.
\end{cnst}
In \cite{con2023optimal}, the authors prove that this code can correct any $n-3$ insdel errors. We will show that it is actually robust against the $(n-3)$-permutation-insdel adversary.

\begin{thm} \label{thm:cnst-k-2}
    For any $n\geq 3$, the code $\text{RS}_{n,2}(\alpha_1, \ldots, \alpha_n)$ defined in \Cref{const:explicit-n-2} is robust against the $(n-3)$-permutation-insdel adversary.
\end{thm}

\begin{proof}
    Let $I,J \in [n]^3$ be two distinct-element sequences. In this scenario, the $V$-matrices are of the form
    \[
    V_{I,J} = 
        \begin{pmatrix}
		1 & \alpha_{I_1} & \alpha_{J_1} \\ 
		1 & \alpha_{I_2} & \alpha_{J_2} \\
		1 & \alpha_{I_3} & \alpha_{J_3} \\
	\end{pmatrix}
    \]
    For $i\in [6]$, let $\delta_i \in \mathcal{A}$ such that
    \begin{align*}
        (\delta_1 + \delta_1^2\gamma, \delta_2 + \delta_2^2\gamma, \delta_3 + \delta_3^2\gamma) &= (\alpha_{I_1}, \alpha_{I_2}, \alpha_{I_3})\\
        (\delta_4 + \delta_4^2\gamma, \delta_5 + \delta_5^2\gamma, \delta_6 + \delta_6^2\gamma) &= (\alpha_{J_1}, \alpha_{J_2}, \alpha_{J_3}) \;.
    \end{align*}
    As in \cite{con2023optimal}, we write $|V_{I,J}|$ as $p_0(\bdelta) + p_1(\bdelta) \gamma + p_2(\bdelta) \gamma^2$ ($\bdelta$ is shorthand for $\delta_1,\ldots,\delta_6$) where,
            \begin{align*}
               p_0(\bm{ \delta}) &= (\delta_1 - \delta_2)(\delta_5 - \delta_6) -  (\delta_2 - \delta_3) (\delta_4-\delta_5) \\
               p_1(\bm{ \delta}) &= (\delta_1 - \delta_2)(\delta_5^2 - \delta_6^2) + (\delta_1^2 - \delta_2^2) (\delta_5 - \delta_6)\\
               &\quad- (\delta_2^2 - \delta_3^2)(\delta_4-\delta_5) - (\delta_2 - \delta_3)(\delta_4^2-\delta_5^2) \\
               p_2(\bm{ \delta}) &= (\delta_1^2 - \delta_2^2)(\delta_5^2 - \delta_6^2) - (\delta_2^2 - \delta_3^2) (\delta_4^2-\delta_5^2) \;.
            \end{align*}
    If $|V_{I,J}| = 0$, it means that $p_i(\bdelta) = 0$ for $i=0,1,2$. From $p_0(\bdelta) = 0$, we have that 
    \begin{equation} \label{eq:p-0}
        (\delta_1 - \delta_2)(\delta_5 - \delta_6) = (\delta_2 - \delta_3) (\delta_4-\delta_5), 
    \end{equation}
    and note that both sides are nonzero as $|\{\delta_1, \delta_2, \delta_3\}| = |\{\delta_4, \delta_5, \delta_6\}| = 3$. From $p_2(\bdelta)$ we get that 
    \begin{equation} \label{eq:p-2}
        (\delta_1 + \delta_2)(\delta_5 + \delta_6) = (\delta_2 + \delta_3)(\delta_4 + \delta_5)\;,
    \end{equation}
    and from $p_1(\bdelta) = 0$,
    \begin{equation} \label{eq:restriction-1}
        \delta_1 + \delta_6 = \delta_3 + \delta_4 \;. 
    \end{equation} 

    Summing this up \eqref{eq:p-0} with \eqref{eq:p-2}, we get that $\delta_5 (\delta_1 - \delta_3) = \delta_2(\delta_4 - \delta_6)$ (note that we used here the fact that the characteristic of the field is $>2$). 
    Now, taking into consideration \eqref{eq:restriction-1}, we have that $\delta_2 = \delta_5$. 
    Note that this implies that we have the minimal equal sub-pair $(I_{\{2\}}, J_{\{2\}})$, which is allowed.
    We will show next that there must be two more minimal equal sub-pairs.

    Plugging $\delta_2 = \delta_5$ in \eqref{eq:p-2}, and taking \eqref{eq:restriction-1} into consideration, we get that, 
    \begin{equation} \label{eq:restriction-2}
            \delta_1 \delta_6 = \delta_3 \delta_4 \;.
    \end{equation}
     By \eqref{eq:restriction-1} and \eqref{eq:restriction-2}, we get that
    $\{\delta_1,\delta_6\}$ and $\{\delta_3,\delta_4\}$ are the solutions to the   quadratic equation $x^2 -(\delta_3 + \delta_4)x + \delta_3 \delta_4 = 0$, and therefore $\{\delta_1,\delta_6\}=\{\delta_3,\delta_4\}$. 
    Now, the case that $\delta_1 = \delta_3$, $\delta_4 = \delta_6$ contradicts the fact that 
    $|\{\delta_1, \delta_2, \delta_3\}| = |\{\delta_4, \delta_5, \delta_6\}| = 3$ and the case that $\delta_1 = \delta_4$ and $\delta_3 = \delta_6$ implies that $(I_{\{1\}}, J_{\{1\}})$ and $(I_{\{3\}}, J_{\{3\}})$ are also minimal equal sub-pairs.
    Thus, the only possible case for $|V_{I,J}| = 0$ is when there are three minimal equal sub-pairs. We conclude that the determinant condition in \Cref{prop:determinant-implies-robustness} holds.

    Thus, \Cref{prop:determinant-implies-robustness} implies that the respective $\text{RS}_{n,2}(\alpha_1, \ldots, \alpha_n)$ is robust against the $n-3$-permutation-insdel adversary.
\end{proof}

\subsection{The general dimension case}
In this section, we show that the construction \cite[Construction 27]{con2023reed} that presents an $\RS$ code that can correct $n-2k+1$ insdel errors is actually also a construction of RS code that is robust against the $(n-2k+1)$-permutation-insdel adversary. 
Specifically, we prove the following.
\begin{thm}\label{thm:cnst-k-gen}
    Let $k$ and $n$ be positive integers where $2k \leq n$. There is a deterministic construction of an $\RS$ code over a field $\Fq$ of order $q = n^{O(k^2 ((2k)!)^2)}$ that is robust against the $n-2k+1$-permutation-insdel adversary.
\end{thm}
We recall the construction:
	\begin{cnst} \cite[Construction 27]{con2023reed}\label{cnst:abc-rs}
		Let $k$ be a positive integer and set $\ell = ((2k)!)^2$.
		Fix a finite field $\Fp$ for a prime $p > k^2 \cdot \ell$ and let $n$ be an integer such that $2k-1< n \leq p$. Let $\Fq$ be a field extension of $\Fp$ of  degree $k^2\cdot \ell$ and let $\gamma\in \Fq$ be such that $\Fq = \Fp(\gamma)$. Hence, each element of $\Fq$ can be represented as a polynomial in $\gamma$, of degree less than $k^2\ell$, over $\Fp$. Define the $\RS$ code by setting 
		 $\alpha_i := (\gamma-i)^{\ell}$ for  $1\leq i\leq n$.
	\end{cnst}

    Another ingredient that we need is a version of the Mason--Stothers theorem that was used in \cite{con2023reed}.

    \begin{thm}[``Moreover part'' of Proposition 5.2 in         \cite{vaserstein2003vanishing}] \label{thm:abc-poly}
		Let $m \geq 2$ and $Y_0(x) = Y_1(x) + \ldots + Y_m(x)$ with $Y_j(x)\in \Fp [x]$. Suppose that $\gcd (Y_0(x), \ldots, Y_m(x)) = 1$, and  that $Y_1(x), \ldots, Y_m(x)$ are linearly independent over $\Fp (x^p)$.\footnote{$\Fp (x^p)$ is the field of rational functions in $x^p$. Namely, its elements are $f(x^p)/g(x^p)$ where $f(x),g(x) \in \Fp[x]$ and $g(x)\not \equiv 0$.} Then, 
		\[
		\deg (Y_0(x)) \leq -\binom{m}{2} + (m-1)\sum_{j=0}^{m} \nu (Y_j(x))\;.
		\]
		where $\nu (Y_j(x))$ is the number of distinct roots of $Y_j(x)$ whose multiplicity is not divisible by $p$.
	\end{thm}

        We are now ready to prove our theorem, and again we shall only indicate which parts of the proof of \cite{con2023reed} change in order not to repeat all the details.
        \begin{prop}
            The $\RS$ code defined in \Cref{cnst:abc-rs} is robust against the $n-2k+1$-permutation-insdel adversary. 
        \end{prop}
        \begin{proof}
             Fix two distinct-element sequences $I,J\in[n]^{2k-1}$ having at most $k$ minimal equal sub-pairs. We follow the proof of \cite[Proposition 28]{con2023reed}. The only place in that proof where the fact that $I$ and $J$ are increasing sequences that agree in at most $k-1$ coordinates is used is the invocation of \cite[Proposition 18]{con2023reed}, which supplies a monomial that is obtained uniquely in the determinant expansion of $V_{I,J}(\bX)$ and therefore has coefficient $\pm1$. In our setting, precisely the same conclusion follows from \Cref{prop:formal-det-1}.

            All the remaining steps of the proof of \cite[Proposition 28]{con2023reed} apply without any change. Thus, the evaluation points in \Cref{cnst:abc-rs} satisfy the hypothesis of \Cref{prop:determinant-implies-robustness}. The result follows from that proposition.
        \end{proof}
\section*{Acknowledgments}
The author would like to thank Yuval Ishai for motivating the author to extend the work \cite{con2023reed}, for pointing out the possible application to fully anonymous secret-sharing, and for providing helpful comments.
The author also thanks Noam Mazor for many insightful discussions and for providing valuable comments.
Finally, the author thanks Amir Shpilka and Zachi Tamo for fruitful discussions throughout the past several years about this topic.
\bibliographystyle{alpha}
\bibliography{ref}

\newcommand{\etalchar}[1]{$^{#1}$}
\begin{thebibliography}{WMSN04}

\bibitem[AGFC07]{abdel2007linear}
Khaled~AS Abdel-Ghaffar, Hendrik~C Ferreira, and Ling Cheng.
\newblock On linear and cyclic codes for correcting deletions.
\newblock In {\em 2007 IEEE International Symposium on Information Theory}, pages 851--855. IEEE, 2007.

\bibitem[BCG{\etalchar{+}}25]{beelen2025reed}
Peter Beelen, Roni Con, Anina Gruica, Maria Montanucci, and Eitan Yaakobi.
\newblock Reed-solomon codes against insertions and deletions: Full-length and rate-$1/2$ codes.
\newblock {\em arXiv preprint arXiv:2501.11371}, 2025.

\bibitem[Bei11]{beimel2011secret}
Amos Beimel.
\newblock Secret-sharing schemes: A survey.
\newblock In {\em International conference on coding and cryptology}, pages 11--46. Springer, 2011.

\bibitem[BGI{\etalchar{+}}25]{bishop2025fully}
Allison Bishop, Matthew Green, Yuval Ishai, Abhishek Jain, and Paul Lou.
\newblock Fully anonymous secret sharing.
\newblock In {\em Annual International Cryptology Conference}, pages 356--389. Springer, 2025.

\bibitem[BGK20]{bogdanov2020threshold}
Andrej Bogdanov, Siyao Guo, and Ilan Komargodski.
\newblock Threshold secret sharing requires a linear-size alphabet.
\newblock {\em Theory of Computing}, 16(1):1--18, 2020.

\bibitem[BS97]{blundo1997anonymous}
Carlo Blundo and Douglas~R Stinson.
\newblock Anonymous secret sharing schemes.
\newblock {\em Discrete Applied Mathematics}, 77(1):13--28, 1997.

\bibitem[CGHL23]{cheng2020efficient}
Kuan Cheng, Venkatesan Guruswami, Bernhard Haeupler, and Xin Li.
\newblock Efficient linear and affine codes for correcting insertions/deletions.
\newblock {\em SIAM Journal on Discrete Mathematics}, 37(2):748--778, 2023.

\bibitem[CGLZ25]{con2024random}
Roni Con, Zeyu Guo, Ray Li, and Zihan Zhang.
\newblock Random reed-solomon codes achieve the half-singleton bound for insertions and deletions over linear-sized alphabets.
\newblock 334:60:1--60:21, 2025.

\bibitem[Con25]{con2025anonymous}
Roni Con.
\newblock Anonymous shamir’s secret-sharing via reed-solomon codes against permutations, insertions, and deletions.
\newblock {\em IEEE Transactions on Information Theory}, 2025.

\bibitem[CR03]{crochemore2003jewels}
Maxime Crochemore and Wojciech Rytter.
\newblock {\em Jewels of stringology: text algorithms}.
\newblock World Scientific, 2003.

\bibitem[CR20]{cheraghchi2020overview}
Mahdi Cheraghchi and Jo{\~a}o Ribeiro.
\newblock An overview of capacity results for synchronization channels.
\newblock {\em IEEE Transactions on Information Theory}, 67(6):3207--3232, 2020.

\bibitem[CST23]{con2023reed}
Roni Con, Amir Shpilka, and Itzhak Tamo.
\newblock Reed solomon codes against adversarial insertions and deletions.
\newblock {\em {IEEE} Trans. Inf. Theory}, 69(5):2991--3000, 2023.

\bibitem[CST24]{con2023optimal}
Roni Con, Amir Shpilka, and Itzhak Tamo.
\newblock Optimal two-dimensional reed-solomon codes correcting insertions and deletions.
\newblock {\em {IEEE} Trans. Inf. Theory}, 70(7):5012--5016, 2024.

\bibitem[DL78]{demillo1978probabilistic}
Richard~A DeMillo and Richard~J Lipton.
\newblock A probabilistic remark on algebraic program testing.
\newblock {\em Information processing letters}, 7(4):193--195, 1978.

\bibitem[DLTX21]{duc2019explicit}
Tai~Do Duc, Shu Liu, Ivan Tjuawinata, and Chaoping Xing.
\newblock Explicit constructions of two-dimensional {R}eed-{S}olomon codes in high insertion and deletion noise regime.
\newblock {\em IEEE Transactions on Information Theory}, 67(5):2808--2820, 2021.

\bibitem[EBG{\etalchar{+}}24]{eldridge2024abuse}
Harry Eldridge, Gabrielle Beck, Matthew Green, Nadia Heninger, and Abhishek Jain.
\newblock $\{$Abuse-Resistant$\}$ location tracking: Balancing privacy and safety in the offline finding ecosystem.
\newblock In {\em 33rd USENIX Security Symposium (USENIX Security 24)}, pages 5431--5448, 2024.

\bibitem[GMO03]{guillermo2003providing}
Mida Guillermo, Keith~M Martin, and Christine~M O'Keefe.
\newblock Providing anonymity in unconditionally secure secret sharing schemes.
\newblock {\em Designs, Codes and Cryptography}, 28:227--245, 2003.

\bibitem[HS21a]{haeupler2021synchronization}
Bernhard Haeupler and Amirbehshad Shahrasbi.
\newblock Synchronization strings and codes for insertions and deletions - {A} survey.
\newblock {\em {IEEE} Trans. Inf. Theory}, 67(6):3190--3206, 2021.

\bibitem[HS21b]{haeupler2017synchronization}
Bernhard Haeupler and Amirbehshad Shahrasbi.
\newblock Synchronization strings: Codes for insertions and deletions approaching the singleton bound.
\newblock {\em Journal of the ACM (JACM)}, 68(5):1--39, 2021.

\bibitem[KKFT08]{kurihara2008new}
Jun Kurihara, Shinsaku Kiyomoto, Kazuhide Fukushima, and Toshiaki Tanaka.
\newblock A new (k, n)-threshold secret sharing scheme and its extension.
\newblock In {\em Information Security: 11th International Conference, ISC 2008, Taipei, Taiwan, September 15-18, 2008. Proceedings 11}, pages 455--470. Springer, 2008.

\bibitem[KOKO02]{kishimoto2002bound}
Wataru Kishimoto, Koji Okada, Kaoru Kurosawa, and Wakaha Ogata.
\newblock On the bound for anonymous secret sharing schemes.
\newblock {\em Discrete Applied Mathematics}, 121(1-3):193--202, 2002.

\bibitem[KT18]{kovavcevic2018codes}
Mladen Kova{\v{c}}evi{\'c} and Vincent~YF Tan.
\newblock Codes in the space of multisets—coding for permutation channels with impairments.
\newblock {\em IEEE Transactions on Information Theory}, 64(7):5156--5169, 2018.

\bibitem[Liu24]{liu2024optimal}
Jingge Liu.
\newblock Optimal {RS} codes and {GRS} codes against adversarial insertions and deletions and optimal constructions.
\newblock {\em IEEE Transactions on Information Theory}, 2024.

\bibitem[LT21]{liu20212}
Shu Liu and Ivan Tjuawinata.
\newblock {On 2-dimensional insertion-deletion Reed-Solomon codes with optimal asymptotic error-correcting capability}.
\newblock {\em Finite Fields and Their Applications}, 73:101841, 2021.

\bibitem[Mit09]{mitzenmacher2009survey}
Michael Mitzenmacher.
\newblock A survey of results for deletion channels and related synchronization channels.
\newblock {\em Probability Surveys}, 6:1--33, 2009.

\bibitem[PCO20]{paskin2020cryptographic}
Anat Paskin-Cherniavsky and Ruxandra~F Olimid.
\newblock On cryptographic anonymity and unpredictability in secret sharing.
\newblock {\em Information Processing Letters}, 161:105965, 2020.

\bibitem[PP92]{phillips1992strongly}
Steven~J Phillips and Nicholas~C Phillips.
\newblock Strongly ideal secret sharing schemes.
\newblock {\em Journal of Cryptology}, 5:185--191, 1992.

\bibitem[RS60]{reed1960polynomial}
Irving~S Reed and Gustave Solomon.
\newblock Polynomial codes over certain finite fields.
\newblock {\em Journal of the society for industrial and applied mathematics}, 8(2):300--304, 1960.

\bibitem[Sch80]{schwartz1980fast}
Jacob~T Schwartz.
\newblock Fast probabilistic algorithms for verification of polynomial identities.
\newblock {\em Journal of the ACM (JACM)}, 27(4):701--717, 1980.

\bibitem[Sha79]{shamir1979share}
Adi Shamir.
\newblock How to share a secret.
\newblock {\em Communications of the ACM}, 22(11):612--613, 1979.

\bibitem[SKSY24]{sabary2024survey}
Omer Sabary, Han~Mao Kiah, Paul~H Siegel, and Eitan Yaakobi.
\newblock Survey for a decade of coding for dna storage.
\newblock {\em IEEE Transactions on Molecular, Biological, and Multi-Scale Communications}, 2024.

\bibitem[SNW02]{safavi2002traitor}
Reihaneh Safavi-Naini and Yejing Wang.
\newblock Traitor tracing for shortened and corrupted fingerprints.
\newblock In {\em ACM workshop on Digital Rights Management}, pages 81--100. Springer, 2002.

\bibitem[SV88]{stinson1988combinatorial}
Douglas~R Stinson and Scott~A Vanstone.
\newblock A combinatorial approach to threshold schemes.
\newblock {\em SIAM Journal on Discrete Mathematics}, 1(2):230--236, 1988.

\bibitem[TSN07]{tonien2007construction}
Dongvu Tonien and Reihaneh Safavi-Naini.
\newblock {Construction of deletion correcting codes using generalized Reed--Solomon codes and their subcodes}.
\newblock {\em Designs, Codes and Cryptography}, 42(2):227--237, 2007.

\bibitem[VW03]{vaserstein2003vanishing}
Leonid~N Vaserstein and Ethel~R Wheland.
\newblock Vanishing polynomial sums.
\newblock {\em Communications in Algebra}, 31(2):751--772, 2003.

\bibitem[WMSN04]{wang2004deletion}
Yejing Wang, Luke McAven, and Reihaneh Safavi-Naini.
\newblock {Deletion correcting using generalized Reed-Solomon codes}.
\newblock In {\em Coding, Cryptography and Combinatorics}, pages 345--358. Springer, 2004.

\bibitem[Zip79]{zippel1979probabilistic}
Richard Zippel.
\newblock Probabilistic algorithms for sparse polynomials.
\newblock In {\em International symposium on symbolic and algebraic manipulation}, pages 216--226. Springer, 1979.

\bibitem[ZSZG26]{zhang2026random}
Yijun Zhang, Yubo Sun, Xiande Zhang, and Gennian Ge.
\newblock Random reed--solomon codes correcting permutations, insertions, and deletions over polynomial-size alphabets.
\newblock {\em arXiv preprint arXiv:2606.22344}, 2026.

\end{thebibliography}

\end{document}